\documentclass[10pt,reqno]{article}
\usepackage{fullpage}

\usepackage[unicode=true]{hyperref}
\usepackage{amsmath}
\usepackage{cite}
\usepackage{amsfonts}
\usepackage{amssymb}
\usepackage{amsthm}
\usepackage{authblk}

\usepackage[pdftex]{color,graphicx}
\usepackage{mypack} 
\usepackage{adjustbox}
\usepackage{soul}
\usepackage{comment}
\usepackage{bbold}
\usepackage{tikz}
\usetikzlibrary{decorations.pathreplacing,angles,quotes}
\usepackage{bbold}

\usepackage{diagbox}
\usepackage{enumerate}

\usetikzlibrary{matrix,arrows,decorations.pathmorphing}
\usepackage{tikz-cd}
\usetikzlibrary{arrows} 

 \usepackage{caption}
\usepackage{subcaption}
 
 \usepackage{pdfpages}

\usepackage{blkarray}

\usepackage{centernot}
\usepackage{mathtools}
\usepackage{stmaryrd}

  \usepackage{arydshln}%dashed horizontal line in table

\definecolor{bluegray}{rgb}{0.4, 0.6, 0.8}
\definecolor{turquoise}{rgb}{0.2, 0.7, 0.6}
\definecolor{hy-green}{rgb}{0.1, 0.5, 0.1}
\definecolor{cyan-custom}{cmyk}{1,0, 0, 0}
\definecolor{red-custom}{cmyk}{0.02 ,1.00, 0.94, 0.00}
\definecolor{gray-custom}{cmyk}{0.00 ,0.00, 0.00, 0.50}

\usepackage{soul}  %strikethrough text: \st{}
 \newcommand{\suchthat}{\;\ifnum\currentgrouptype=16 \middle\fi|\;}

\title{{The degenerate vertices of the $2$-qubit $\Lambda$-polytope and their update rules}} 
	\author{Selman Ipek\footnote{selman.ipek@bilkent.edu.tr} }
	\author{Cihan Okay\footnote{cihan.okay@bilkent.edu.tr}}
	\affil{Department of Mathematics, Bilkent University, Ankara, Turkey}

\begin{document}
  \maketitle  
  
%\noindent{\textbf{From}: Selman Ipek}\\
%\noindent{\textbf{To}: Cihan Okay}

%\begin{abstract}
%
%\end{abstract}

\begin{abstract}
{Recently, a class of objects, known as $\Lambda$-polytopes, were introduced for classically simulating universal quantum computation with magic states. In $\Lambda$-simulation, the probabilistic update of $\Lambda$ vertices under Pauli measurement {yields} dynamics consistent with quantum mechanics. Thus, an important open problem in the study of $\Lambda$-polytopes is characterizing its vertices and determining their update rules. In this paper, we obtain and describe the update of all degenerate vertices of $\Lambda_{2}$, the $2$-qubit $\Lambda$ polytope. Our approach exploits the fact that $\Lambda_{2}$ projects to a well-understood polytope $\MP$ consisting of distributions on the Mermin square scenario. More precisely, we study the ``classical" polytope $\overline{\MP}$, which is $\MP$ intersected by the polytope defined by a set of Clauser-Horne-Shimony-Holt (CHSH) inequalities. Owing to a duality between CHSH inequalities and vertices of $\MP$ we utilize a streamlined version of the double-description method for vertex enumeration to obtain certain vertices of $\overline{\MP}$.}
\end{abstract}

\tableofcontents

\section{Introduction}\label{sec:intro}

{Quantum computation with magic states (QCM) is a model of universal quantum computation in which the computational primitives include arbitrary quantum states evolving according to update under Clifford unitaries and Pauli measurements {\cite{magic}}. Recently a quasiprobability-based algorithm for classically simulating QCM with qubits was introduced by Zurel, et al. \cite{zurel2020hidden} based on a family of convex polytopes called $\Lambda$-polytopes. In contrast to many quasiprobability-based approaches, in $\Lambda$-simulation all computational primitives of QCM are positively represented, including the representation of all quantum states, which for any number $n$ of qubits always lie inside $\Lambda_{n}$, the $n$-qubit $\Lambda$ polytope.
}

{As part of a broader effort to understand the computational power of quantum computers an open problem in the area of classical simulation of quantum circuits is to characterize the efficiency (or lack thereof) of $\Lambda$-simulation. Progress along these lines was recently made in \cite{zurel2023simulating}. One of the remaining open problems in analyzing the complexity of $\Lambda$-simulation, however, is in characterizing the hardness of the vertex enumeration problem for $\Lambda$-polytopes, which is a necessary precondition for running the simulation and the key bottleneck one encounters in practice. Indeed, there are no known polynomial time algorithms for solving the vertex enumeration problem for general polytopes (see e.g., \cite{kaibel2003some}) and due to a classic result by Khachiyan, et al.  \cite{khachiyan2009generating} the problem is known to be \textbf{$\text{NP}$}-{hard} for (unbounded) polyhedra unless $\mathbf{P}=\mathbf{NP}$.}

{Here we take steps to better understand this problem for $\Lambda$-polytopes by studying the vertices of $\Lambda_{2}$, the $2$-qubit $\Lambda$ polytope.} In previous work \cite{okay2022mermin} it was shown that $\Lambda_{2}$ was a sub-polytope of the nonsignaling polytope corresponding to the $(2,3,2)$ Bell scenario{, which we denote by $\NS$}; see e.g., \cite{jones2005interconversion}. {W}e refine this result by showing that several new types of vertices (under the action of the $2$-qubit Clifford group) can be captured by considering the intersection of $\Lambda_{2}$ with the polytope $\CL$, which is the ``classical" part of $\NS$; i.e., the Bell polytope of the $(2,3,2)$ Bell scenario.
Here is a diagram of the polytopes involved in our {analysis:} 
%lifting construction:
$$
\begin{tikzcd}
 & \NS & \\
\Lambda_2 \arrow[ur] & & \CL \arrow[ul] \\
&\Lambda_2\cap \CL \arrow[ul] \arrow[ur]&
\end{tikzcd}
\begin{tikzcd}\;\;\;\;
\phantom{a}\arrow[r,"\pi"] & \phantom{a}
\end{tikzcd}\;\;\;\;
\begin{tikzcd}
 &  & \\
\MP   & & {\overline\CL}  \\
& {\overline\MP} \arrow[ul] \arrow[ur]&
\end{tikzcd}
$$
The non-signaling polytope $\NS$ is contained in $\RR^{15}$ and we can see $\Lambda_2$ as a subpolytope there by considering the Pauli expectation values. The map $\pi$ projects onto the non-local Pauli operators, that is, onto $\RR^9$. The Mermin polytope $\MP$ is given by $\pi(\Lambda_2)$ and the projected classical polytope is $\overline\CL=\pi(\CL)$. The strategy taken in this paper is to describe vertices of $\Lambda_2$ by lifting the vertices of a polytope in the image of $\pi$. 
It is computationally known that the $2$-qubit $\Lambda$ polytope has $8$ types of vertices ($T_i,\; i=1,2,\cdots,8$) each consisting of the orbits under the action of the $2$-qubit Clifford group.
However, the vertices of $\Lambda_n$ for $n\geq 3$ are not known, even computationally. {Only certain subclasses are known to date. One such class of vertices are called closed noncontextual, or cnc vertices \cite{raussendorf2020phase}.} For $\Lambda_2$ these cnc vertices cover the first two types $T_1$ and $T_2$:
$$  
T_1:
\begin{tabular}{ c|c c c } 
 $1$ & $0$ & $0$ & $0$\\
 \hline 
 $0$ & $1$ & $1$ & $1$ \\ 
 $1$ & $0$ & $0$ & $0$  \\
 $1$ & $0$ & $0$ & $0$
\end{tabular}\notag
\;\;\;\;
\;\;\;\;
\;\;\;\;
\;\;\;\;
T_2:
\begin{tabular}{ c|c c c } 
 $1$ & $0$ & $0$ & $1$\\
 \hline 
 $0$ & $1$ & $1$ & $0$ \\ 
 $0$ & $1$ & $-1$ & $0$  \\
 $1$ & $0$ & $0$ & $1$
\end{tabular}\notag
$$
Another class of vertices valid for all $\Lambda_{n}$ was recently discovered in \cite{zurel2023simulation} based on a connection between line graphs and Majorana Fermions. In the case of $\Lambda_{2}$ this corresponds to a $T_{4}$ vertex. Together with the $T_{8}$ vertices a distinguishing feature of $T_{4}$ vertices for $2$-qubits is that they are non-degenerate \cite{luenberger1984linear}:{
$$  
T_4:
\begin{tabular}{ c|c c c } 
 $1$ & $1/2$ & $-1/2$ & $1/2$\\
 \hline 
 $0$ & $1/2$ & $1/2$ & $1/2$ \\ 
 $-1/2$ & $0$ & $0$ & $0$  \\
 $0$ & $1/2$ & $1/2$ & $1/2$
\end{tabular}\notag
\;\;\;\;
\;\;\;\;
\;\;\;\;
\;\;\;\;
T_8:
\begin{tabular}{ c|c c c } 
 $1$ & $3/4$ & $1/2$ & $1/2$\\
 \hline 
 $3/4$ & $1/2$ & $3/4$ & $1/4$ \\ 
 $1/4$ & $1/2$ & $-1/4$ & $-1/4$  \\
 $3/4$ & $1/2$ & $1/4$ & $3/4$
\end{tabular}\notag
$$ 
}

The vertex enumeration problem for $\MP$ is solved in \cite{okay2022mermin}. There are two types ($\bar T_1$ and $\bar T_2$) of vertices up to the action of the combinatorial automorphism group.
It turns out that $T_1$ and $T_2$ vertices map to $\bar T_1$ and $\bar T_2$ vertices under $\pi$ map. In other words, the vertices of $\MP$ give us information about the first two types of vertices of $\Lambda_2$.  
In this paper, we explore the possibility of whether other types of vertices can be captured by a polytope in $\RR^9$. We explore this question by studying the vertices of the intersection polytope
$$
\overline\MP  =  \MP \cap \overline\CL.
$$
By computer computation, we find that this polytope has $10$ types of vertices ($\tau_i, i=1,2,\cdots,10$).  
Using the combinatorial structure of the Mermin polytope studied in \cite{okay2022mermin} and a version of the Double Description method (DD) of polytope theory \cite{fukuda2005double}, we described four types: $\tau_3,\tau_5,\tau_6,\tau_7$. {Our implementation of the DD method is facilitated by a duality between $\bar T_{2}$ vertices and the facets of $\overline{\text{CL}}$ (see Fig.~(\ref{fig:duality})) which motivated Proposition~\ref{pro:DD for intersection of pair of polytopes}. This duality was, in fact, first observed by Howard and Vala in \cite{howard2012nonlocality} in relation to characterizing nonlocality (violation of a Bell inequality) as a necessary pre-condition for universal quantum computation.
}
\begin{figure}[h!]
\centering
\includegraphics[width = 0.3\linewidth]{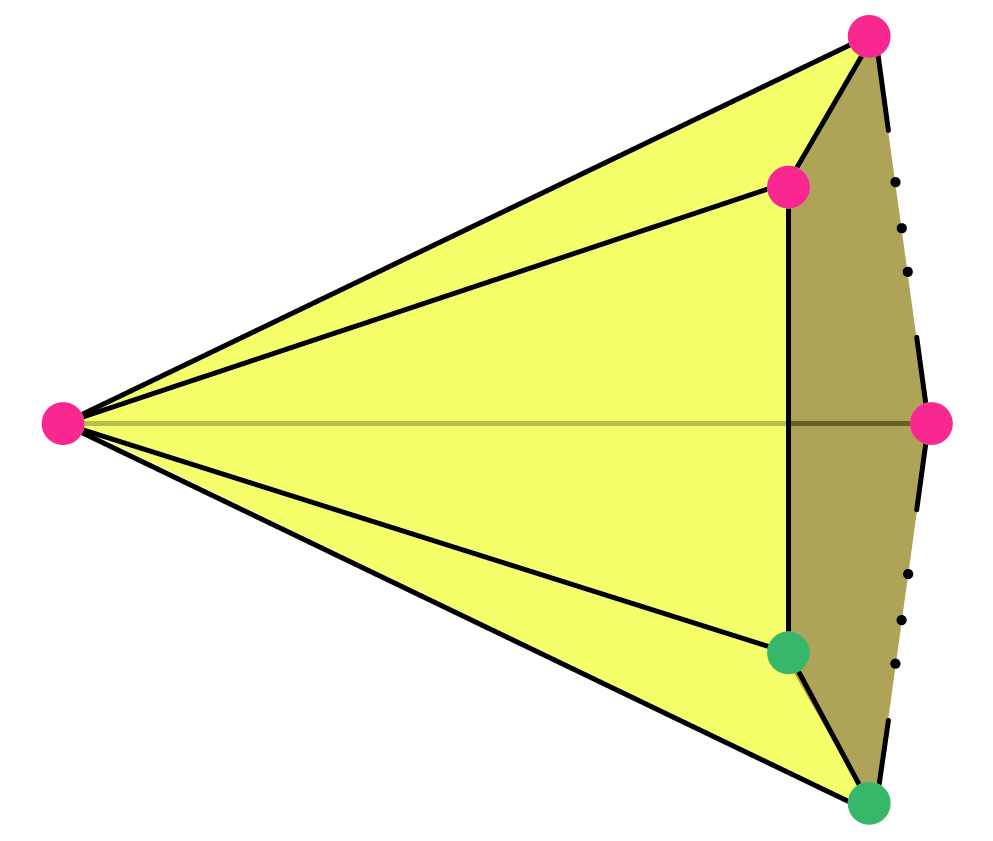}
\caption{{An illustration of the duality between $\bar T_{2}$ vertices (pink) and facets of $\overline{\text{CL}}$ (shaded (hyper)plane). The latter are (essentially) Clauser-Horne-Shimony-Holt (CHSH) facets. Moreover, all $\bar T_{1}$ (green) and $\bar T_{2}$ (pink) neighbors of any $\bar T_{2}$ are tight at the CHSH facet dual to it.}}
\label{fig:duality}
\end{figure}

%%. More explicitly, they work with the Clifford polytope \cite{buhrman2006new}, which is the polar dual \cite{ziegler2012lectures} of $\MP$ and they establish that a single-qubit channel 
%In their setting, they work with the Clifford polytope \cite{buhrman2006new}, which is the polar dual \cite{ziegler2012lectures} of $\MP$ and they establish that

%\comm{connect to Howard et al and talk about the dual relation between CHSH and $\bar T_2$ to motivate our main polytope Proposition}

Having determined some of the vertices of $\overline\MP$ we provide convex decompositions of these vertices in $\overline \CL$. We show that they can be written as uniform mixtures of a small number of $\pi(d)$ where $d$ is a deterministic vertex of $\CL$. Each $\pi(d)$ can lift to a line segment $pd+(p-1)\bar d$ where $p\in [0,1]$ and $\bar d$ is the deterministic vertex whose outcomes are flipped for each party.
For suitable choices of $p\in \set{0,1}$ we obtain the following vertices of $\Lambda_2$:
\begin{equation}\label{eq:vertices-lambda2}
\begin{aligned}
T_3:&
\begin{tabular}{ c|c c c } 
 $1$ & $-1/2$ & $-1/2$ & $1/2$\\
 \hline 
 $-1/2$ & $1$ & $0$ & $0$ \\ 
 $0$ & $1/2$ & $-1/2$ & $-1/2$  \\
 $1$ & $-1/2$ & $-1/2$ & $1/2$
\end{tabular}
&\;\;\;\;\;\;&
T_5:&
\begin{tabular}{ c|c c c } 
 $1$ & $-1/2$ & $-1/2$ & $-1/2$ \\
 \hline 
 $-1/2$ & $1$ & $0$ & $0$ \\ 
 $1/2$ & $0$ & $-1$ & $0$  \\
 $-1/2$ & $0$ & $0$ & $1$
\end{tabular}
\\
\\
T_6:&
\begin{tabular}{ c|c c c } 
 $1$ & $-1$ & $1/2$ & $1/2$ \\
 \hline 
 $-1/2$ & $1/2$ & $0$ & $0$ \\ 
 $-1/2$ & $1/2$ & $-1$ & $0$  \\
 $0$ & $0$ & $-1/2$ & $1/2$
\end{tabular}
&&
T_7:&
\begin{tabular}{ c|c c c } 
 $1$ & $-2/3$ & $1/3$ & $2/3$ \\
 \hline 
 $-1/3$ & $2/3$ & $1/3$ & $0$ \\ 
 $-2/3$ & $1/3$ & $-2/3$ & $-1/3$  \\
 $2/3$ & $-1/3$ & $0$ & $1$
\end{tabular}
\end{aligned}
\end{equation}
Our approach provides a decomposition of these types of vertices as uniform mixtures of deterministic vertices. From this decomposition the update rules can be computed by analyzing the updates of the deterministic vertices.
Together with the $T_1$ and $T_2$ vertices this covers all the degenerate vertices of $\Lambda_2$. {For the remaining non-degenerate types, $T_4$ and $T_8$, although {there are vertices in these orbits that} fall into $\overline\MP$ under the image of $\pi$, they do not map to a vertex, and thus cannot be captured by our method.}

{The structure of the paper is as follows. We introduce the polytopes relevant for our analysis in Section~\ref{sec:lambda}. This includes a description of our vertex enumeration strategy, which builds off of basic polytope-theoretic notions introduced in Appendix~\ref{sec:DDM-VEP}. In Section~\ref{sec:Structure of the Mermin polytope} we review results related to Mermin polytopes that will be needed in our application of the DD method, which we perform in Section~\ref{sec:classical-mermin}. We lift vertices of $\overline{\MP}$ to vertices of $\Lambda_{2}$ in Section~\ref{sec:lifting} and conclude with some closing remarks in Section~\ref{sec:conclusion}.}

%\comm{finish with the organization of the paper}

\paragraph{Acknowledgments.}
This work is supported by the Digital Horizon Europe project FoQaCiA, GA no. 101070558. The authors also want to acknowledge support from the US Air Force Office of Scientific Research under award number FA9550-21-1-0002.

\section{$\mathbf{\Lambda}$-polytopes}\label{sec:lambda}

The $n$-qubit Pauli operators acting on $\hH=(\CC^2)^{\otimes n}$ are given by 
$$
A_1\otimes A_2 \otimes \cdots \otimes A_n,\;\;\;\;A_i\in \set{\one,X,Y,Z}
$$
where $\one,X,Y,Z$ are the $2\times 2$ Pauli matrices: 
%\comm{matrix notation to be unified}:
$$
\one = \begin{bmatrix}
1 & 0\\
0 & 1
\end{bmatrix} \;\;\;\; X =\begin{bmatrix}
0 & 1\\
1 & 0
\end{bmatrix}\;\;\;\; Y = \begin{bmatrix}
0 & -i\\
i & 0
\end{bmatrix}\;\;\;\;
Z=\begin{bmatrix}
1 & 0\\
0 & -1
\end{bmatrix}.  
$$
These operators constitute an operator basis of the space $\Herm(\hH)$ of hermitian operators.
The $n$-qubit Pauli group $P_n$ is defined to be the subgroup of the unitary group $U(\hH)$ generated by the Pauli operators. 
The $n$-qubit Clifford group is the quotient 
$$
\Cl_n =\frac{N(P_n)}{\Span{e^{i\pi \alpha}\one:\, \alpha\in \RR}} 
$$
where $N(P_n)$ is the normalizer of $P_n$ in $U(\hH)$.
We will call a maximal abelian subgroup $S\subset P_n$ a stabilizer subgroup if every element squares to the identity operator.  
A stabilizer state is the common eigenstate of the operators in a stabilizer subgroup{, which we will denote using the associated projector.} 
The $n$-qubit $\Lambda$-polytope is defined by
\begin{equation}\label{eq:Lambda n}
\Lambda_n = \set{{A}\in \Herm(\hH):\, \Tr({A})=1,\;\; \Tr({A}\Pi)\geq 0,\; \forall \text{ $n$-qubit stabilizer state } \Pi}
\end{equation}
%\si{[[@Cihan: Why not use $A$ for Hermitian here?]]}\\
%$\Lambda_n$ consists of $X\in \Herm(\hH)$ such that
% $\Tr(X)=1$ and
 %$\Tr(X\Pi)\geq 0$ for every  stabilizer state $\Pi$.
 For notation and background on polytope theory we refer to Section \ref{sec:polytope-theory}.

\subsection{The $2$-qubit $\Lambda$-polytope}
For $2$-qubits we partition the set of Pauli operators (ignoring the identity $I\otimes I$) into a local and non-local parts: 
\begin{itemize}
\item Local $2$-qubit Pauli operators
\begin{equation}
\label{eq:Pauli local}
\set{X\otimes \one, Y\otimes \one,Z\otimes \one, \one\otimes X, \one\otimes Y,\one\otimes Z}.
\end{equation}
\item Nonlocal $2$-qubit Pauli operators
\begin{equation}
\label{eq:Pauli non-locali}
\set{X\otimes X, X\otimes Y,X\otimes Z,Y\otimes X,Y\otimes Y,Y\otimes Z,Z\otimes X,Z\otimes Y,Z\otimes Z}.
\end{equation}
\end{itemize}
A $2$-qubit stabilizer subgroup consists of $4$ elements.
In total there are $60$ stabilizer subgroups. Ignoring the $\pm1$ signs there are essentially $15$ of them.
That is, a stabilizer subgroup can be written as
$$
S=\set{I,(-1)^{\eta(A)} A, (-1)^{\eta(B)}B,(-1)^{\eta(A)+\eta(B)}AB }.
$$ 
where $\eta:\set{A,B}\to \ZZ_2$ is a function and $\set{I,A,B,AB}$ is one of the following subgroups:
\begin{itemize}
\item Local stabilizer subgroups
\begin{equation}
\label{eq:local}
\begin{aligned}
&\{I\otimes I,X\otimes I,I\otimes X,X\otimes X\}, \{I\otimes I,X\otimes I,I\otimes Y,X\otimes Y\}, \{I\otimes I,X\otimes I,I\otimes Z,X\otimes Z\} 
\\
&\{I\otimes I,Y\otimes I,I\otimes X,Y\otimes X\}, \{I\otimes I,Y\otimes I,I\otimes Y,Y\otimes Y\}, \{I\otimes I,Y\otimes I,I\otimes Z,Y\otimes Z\}.
\\
&\{I\otimes I,Z\otimes I,I\otimes X,Z\otimes X\}, \{I\otimes I,Z\otimes I,I\otimes Y,Z\otimes Y\}, \{I\otimes I,Z\otimes I,I\otimes Z,Z\otimes Z\}.
\end{aligned} 
\end{equation}
\item Non-local stabilizer subgroups
\begin{equation}
\label{eq:non-local}
\begin{aligned}
&\{I\otimes I,X\otimes Y,Y\otimes X,Z\otimes Z\}, \{I\otimes I,Y\otimes Z,Z\otimes Y,X\otimes X\}, \{I\otimes I,Z\otimes X,X\otimes Z,Y\otimes Y\},\\
&\{I\otimes I,X\otimes Y,Z\otimes X,-Y\otimes Z\}, \{I\otimes I,Y\otimes X,X\otimes Z,-Z\otimes Y\}, \{I\otimes I,Z\otimes Z,Y\otimes Y,-X\otimes X\}.
\end{aligned} 
\end{equation}
\end{itemize}
There is a one-to-one correspondence between stabilier states and stabilizer subgroups.
The stabilizer state corresponding to $S$ is given by the projector
$$
\Pi_S =  \frac{1}{4} \sum_{A\in S}A.
$$
The $2$-qubit $\Lambda$-polytope {defined in Eq.~(\ref{eq:Lambda n}), where $n=2$,} can be described as a polytope 
\begin{equation}
\label{eq:H description}
P(M,{-I_{60\times 1}})=\set{x\in \RR^{15}:\, Mx\geq {-I_{60\times 1} }}
\end{equation}
where $I_{60\times 1}$ is the column matrix consisting of $1$'s and 
 $M$ is a $60\times 15$ matrix, whose rows are indexed by stabilizer subgroups (states) and columns indexed by the Pauli operators, defined by
$$
M_{S,A} = \Tr(\Pi_S A).
$$  
It will be convenient to represent a Hermitian operator {$A$} as a tableau
%\begin{eqnarray}
%\begin{tabular}{ c|c c c } 
% $\Span{A}_{II}$ & $\Span{A}_{IX}$ & $\Span{A}_{IY}$ & $\Span{A}_{IZ}$ \\
% \hline 
% $\Span{A}_{XI}$ & $\Span{A}_{XX}$ & $\Span{A}_{XY}$ & $\Span{A}_{XZ}$ \\ 
% $\Span{A}_{YI}$ & $\Span{A}_{YX}$ & $\Span{A}%_{YY}$ & $\Span{A}_{YZ}$  \\
% $\Span{A}_{ZI}$ & $\Span{A}_{ZX}$ & $\Span{A}_{ZY}$ & $\Span{A}_{ZZ}$
%\end{tabular}%
%\label{eq:expectation-tableau}
%\end{eqnarray}
\begin{eqnarray} 
\begin{tabular}{ c|c c c } 
 $x_{II}$ & $x_{IX}$ & $x_{IY}$ & $x_{IZ}$ \\
 \hline 
 $x_{XI}$ & $x_{XX}$ & $x_{XY}$ & $x_{XZ}$ \\ 
 $x_{YI}$ & $x_{YX}$ & $x_{YY}$ & $x_{YZ}$  \\
 $x_{ZI}$ & $x_{ZX}$ & $x_{ZY}$ & $x_{ZZ}$
\end{tabular}%
\label{eq:expectation-tableau}
\end{eqnarray}
where {$x_{BC}$ is the expectation $\Span{A}_{BC}=\Tr(A (B\otimes C))$}.

\subsection{Vertex enumeration strategy} 
 
Our goal is to describe the vertices of $\Lambda_2$ from the $H$-description $P(M,{-}I_{60\times 1})$ given in Eq.~(\ref{eq:H description}).  
Let us write 
\begin{equation}\label{eq:projection}
\pi:\RR^{15}\to \RR^9
\end{equation}
for the map that projects onto the non-local coordinates, i.e., those indexed by $\Span{{A}}_{BC}$ in Eq.~(\ref{eq:expectation-tableau}) where ${B,C}\in \set{X,Y,Z}$.
We begin by decomposing the polytope into two parts:
\begin{itemize}
\item Local $\Lambda_2$-polytope: $\Lambda_2^\loc$ is defined by $P(M^\loc,{-}I_{36\times 1})$  where $M^\loc$ is the $36\times 15$ matrix obtained from $M$ by keeping the rows indexed by local stabilizer states. {As shown in \cite{okay2022mermin} t}his polytope can be identified with the non-signaling polytope $\NS$ of the $(2,3,2)$ Bell scenario \cite{jones2005interconversion}:
\begin{equation}
\label{eq:Lambda local with NS}
\NS\cong \Lambda_2^\loc.
\end{equation}
The identification between $\Lambda^\loc_2$ and $\NS$ is as follows: The two parties Alice and Bob each perform the Pauli measurements $X,Y,Z$ to obtain an outcome in $\ZZ_2{=\set{0,1}}$. Interpreting the expectations describing the points in $\Lambda^\loc_2$ as probability distributions on pairs of measurements $A\otimes I$ for Alice and $I\otimes B$ for Bob, where $A,B\in \set{X,Y,Z}$, specifies a point in $\NS$. %This mapping turns out to be an isomorphism of the polytopes.

\item Non-local $\Lambda_2$-polytope: $\Lambda_2^\nloc$ is defined by $P(M^\nloc,{-}I_{24\times 1})$  where $M^\nloc$ is the $24\times 15$ matrix obtained from $M$ by keeping the rows indexed by non-local stabilizer states. The image of this polytope under the projection map $\pi$ is the Mermin polytope\footnote{{In \cite{okay2022mermin} this polytope is denoted by $\MP_1$.}} $\MP$ studied in \cite{okay2022mermin}:
\begin{equation}
\label{eq:Lambda nonlocal with MP}
\MP \cong \pi(\Lambda_2^\nloc).
\end{equation}
More explicitly, let $\bar M^\nloc$ be the matrix obtained from $M^\nloc$ by keeping the columns indexed by non-local Pauli operators. Then 
\begin{equation}
\label{eq:MP1 polytope}
\MP = P(\bar M^\nloc,{-}I_{24\times 1}).
\end{equation}
We will consider the 
%symmetry 
{subgroup}
\begin{equation}
\label{eq:G1}
{G} = \Span{\Cl_1\times \Cl_1,\SWAP}
\end{equation}
of the $2$-qubit Clifford group   generated by the single qubit Clifford unitaries and the $\SWAP$ operator that permutes the two tensor factors. This group acts on $\MP$ by combinatorial automorphisms.
\end{itemize} 

\begin{figure}
\centering
\begin{subfigure}{.33\textwidth}
\[\begin{tabular}{ c|c c c } 
 $1$ & $1$ & $0$ & $0$ \\
 \hline 
 $1$ & $1$ & $0$ & $0$ \\ 
 $0$ & $0$ & $0$ & $0$  \\
 $0$ & $0$ & $0$ & $0$
\end{tabular},\]
\caption{NN}
\label{fig:inequalities-nonnegative}
\end{subfigure}%
\begin{subfigure}{.33\textwidth}
\[\begin{tabular}{ c|c c c } 
 $2$ & $0$ & $0$ & $0$ \\
 \hline 
 $0$ & $-1$ & $-1$ & $0$ \\ 
 $0$ & $-1$ & $1$ & $0$  \\
 $0$ & $0$ & $0$ & $0$
\end{tabular},\]
\caption{CHSH}
\label{fig:ineq-chsh}
\end{subfigure}%
\begin{subfigure}{.33\textwidth}
\[\begin{tabular}{ c|c c c } 
 $4$ & $1$ & $0$ & $-1$ \\
 \hline 
 $-1$ & $-1$ & $1$ & $1$ \\ 
 $0$ & $-1$ & $0$ & $-1$  \\
 $-1$ & $-1$ & $-1$ & $1$
\end{tabular}\]
\caption{FCG}
\label{fig:inequalities-fcg}
\end{subfigure}
\caption{Representatives of the three types of inequalities describing {$\text{CL}$}}
\label{fig:232-inequalities}
\end{figure}

There are special vertices in the non-signaling polytope {$\NS$} called the deterministic vertices. We can represent them in a tableau:
% as
\begin{eqnarray}
\begin{tabular}{ c|c c c } 
$1$ & $(-1)^{s_{0}}$ & $(-1)^{s_{1}}$ & $(-1)^{s_{2}}$ \\
 \hline 
 $(-1)^{r_{0}}$ & $(-1)^{r_{0}+s_{0}}$ & $(-1)^{r_{0}+s_{1}}$ & $(-1)^{r_{0}+s_{2}}$ \\ 
 $(-1)^{r_{1}}$ & $(-1)^{r_{1}+s_{0}}$ & $(-1)^{r_{1}+s_{1}}$ & $(-1)^{r_{1}+s_{2}}$  \\
 $(-1)^{r_{2}}$ & $(-1)^{r_{2}+s_{0}}$ & $(-1)^{r_{2}+s_{1}}$ & $(-1)^{r_{1}+s_{1}}$
\end{tabular}%
\label{eq:deterministic-vertices}
\end{eqnarray}
where we 
{write}
%use the notation that 
$r_{i}$,  $i=0,1,2$, for the outcome of the Pauli measurements $X\otimes I, Y\otimes I, Z\otimes I$, respectively. {Similarly, we write} 
%for 
$s_i$ {for} the measurement outcomes of Bob.  
Let $D$ denote the set of deterministic vertices. The classical polytope is defined to be the convex hull of the deterministic vertices:
\begin{equation}
\label{eq:def-c232}
{\CL}= \conv(D).
\end{equation}
The $H$-representation of ${\CL}$ will be important. 
Following the  
results of Froissart \cite{froissart1981constructive} as well as Collins and Gisin \cite{collins2004relevant}, there are three orbits of facets under the automorphism group of ${\CL}$: $36$ non-negativity (NN) inequalities; $72$ Clauser--Horne--Shimony--Holt (CHSH) inequalities; $576$ Froissart, Collins--Gisin (FCG) inequalities. See Fig.~(\ref{fig:232-inequalities}).

Now, we introduce a polytope whose $V$-description will play a key role in the vertex enumeration problem for $\Lambda_2$:
\begin{equation}
\label{eq:barMP}
\overline\MP = \MP \cap {\overline \CL}
\end{equation}
where ${\overline \CL}=\pi({\CL})$, the  projection of the classical polytope {under the  map in (\ref{eq:projection})}. The projection of a FCG inequality can be written as a convex mixture of CHSH inequalities \cite{howard2012nonlocality}.  
Therefore the $H$-description of ${\overline \CL}$ consists of {$72$ CHSH inequalities}  together with 
%the 
%non-negativity
{$18$ NN} inequalities.
% ($18$). 
The former set can be obtained by the action of ${G}$ defined in Eq.~(\ref{eq:G1}) on the representative given in Fig.~(\ref{fig:ineq-chsh}). 
The latter are of the form  $-1\leq {\Span{A}_{BC}} \leq 1$.  
To describe vertices of $\overline \MP$ we will use
{modified version of the double description (DD) algorithm \cite{fukuda2005double}
given in Proposition \ref{pro:DD for intersection of pair of polytopes}} with 
$$
P_1={\MP}\;\;\;\text{ and }\;\;\; P_2={\overline \CL}
$$
and then we will lift those vertices to vertices of $\Lambda_2$. First we will show that $(P_1,P_2)$ satisfies the required assumptions of the proposition. For this we need an explicit description of the combinatorial structure of $\MP$, as described in \cite{okay2022mermin}, and how the CHSH inequalities of ${\overline \CL}$ relate to this structure.

\section{Mermin polytope} 
\label{sec:Structure of the Mermin polytope} 
 
We begin with recalling the symplectic structure of the Pauli group. For concreteness we will do this for the $2$-qubit case. 
We 
%will 
denote the Pauli matrices as follows
$$
T_a = \left\lbrace
\begin{array}{ll}
\one & a=(0,0) \\
X  & a=(0,1) \\
Y & a=(1,1) \\
Z & a=(1,0).
\end{array}
\right.
$$
Consider the additive group $E=\ZZ_2^2\times \ZZ_2^2$.
A $2$-qubit Pauli operator, a tensor product of two single qubit Pauli matrices, can be written as
$$
T_a = T_{a_1} \otimes T_{a_2}
$$
where $a=(a_1,a_2)\in E$. These operators satisfy two important relations:
\begin{itemize}
\item The commutation relation:
$$
T_a T_b = (-1)^{\omega(a,b)} T_b T_a
$$
where $\omega:E\times E\to \ZZ_2$ is a symplectic form defined by
$$
\omega =
\begin{bmatrix}
0 & I\\
I & 0
\end{bmatrix}.
$$
{We say $a,b\in E$ (anti-)commute if $\omega(a,b)=0$ ($\omega(a,b)=1$).}

\item The product relation:
$$
T_a T_b = (-1)^{\beta(a,b)} T_{a+b}
$$
for $a,b\in E$ with $\omega(a,b)=0$, where $\beta:E\times E\to \ZZ_2$ is a function.
\end{itemize} 
A subspace $I\subset {E}$ is called {\it isotropic} if $\omega(a,b)=0$ for all $a,b\in I$.
The set $E$ {can be decomposed} into {\it local} and {\it non-local} parts corresponding to the decomposition of Pauli operators given in Eq.~(\ref{eq:local}) and Eq.~(\ref{eq:non-local}):
$$
E = \set{0} \sqcup E^\loc \sqcup E^\nloc.
$$
For the description of the Mermin polytope we need the non-local part. {As observed in \cite{Coho,okay2022mermin}  i}t is convenient to organize them into a torus as in Fig.~(\ref{fig:mermin-scenario}). In this description non-local stabilizer subgroups in Eq.~(\ref{eq:non-local}) correspond to the triangles in the torus. {We will rely on this representation for counting arguments related to non-local isotropic subspaces.} 

 \begin{figure}[h!] 
  \centering
  \includegraphics[width=.3\linewidth]{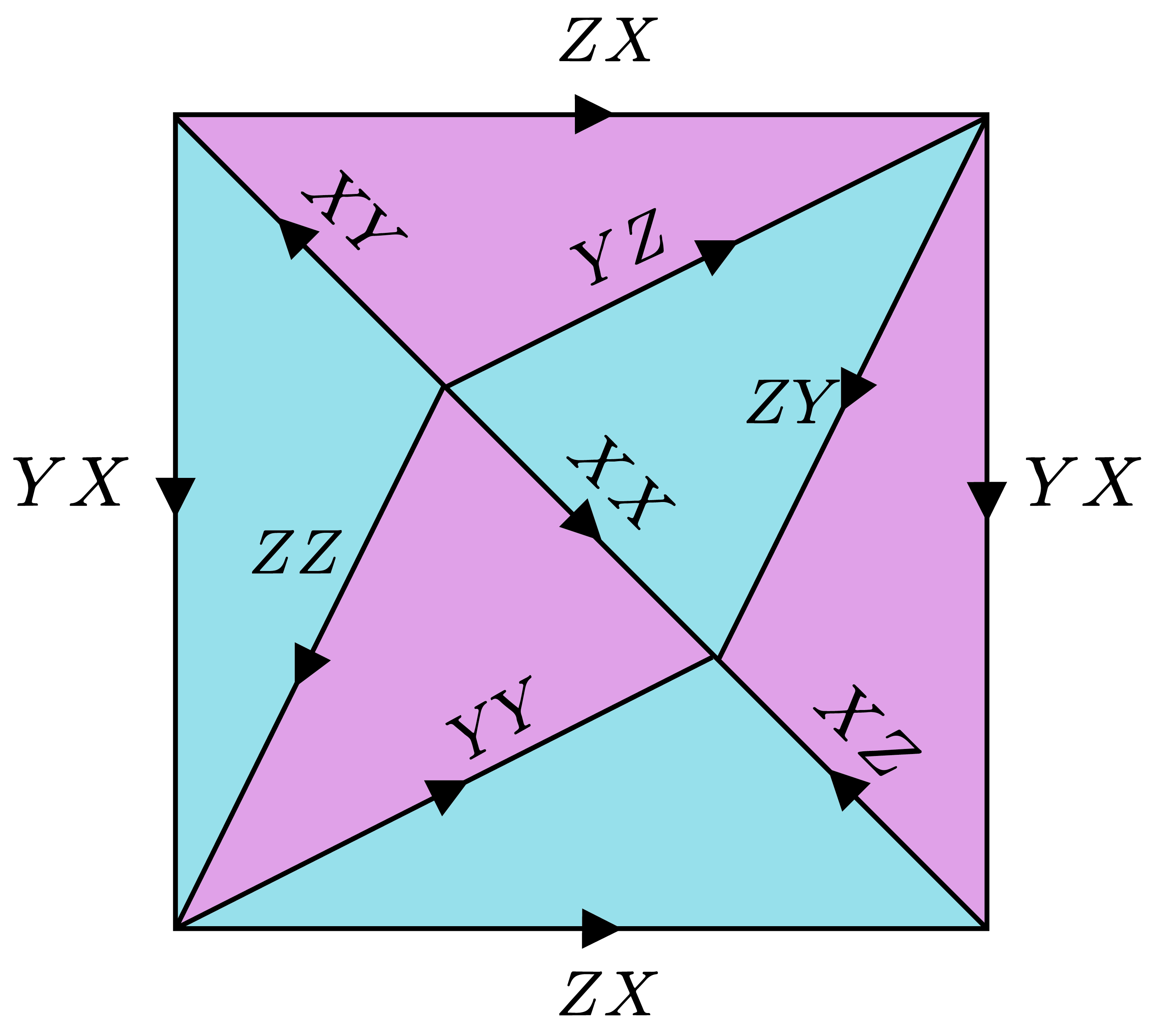}
\caption{ Mermin torus consists of $6$ triangles. The pair of edges on the boundary with the same label are identified. Blue triangles correspond to $\beta=0$ and pink triangles to $\beta=1$.
%Non-local Pauli operators $AB$, where $A,B\in \set{X,Y,Z}$, are used to label the edges. 
%(b) Two neighbors $(\Omega,r)$ of the canonical $T_2$ vertex $q_0$ corresponding to $(\Omega_0,s)$. The value assignment is indicated by a blue ($s(a)=0$) or a red edge $(s(a)=1)$. The neighbors are described by a loop determined by a cnc set. The stabilizer of $q_0$ in $\co{G}$ acts transitively on those neighbors connected via the complement of a $T_1$ (or $T_2$) cnc set. The function $\varphi:l_\Omega \to \ZZ_2$ is given by $\varphi(a)=r(a)$ if $a\in \Omega$ or $\varphi(a)=s(a)+1$ if $a\in \Omega_0$.
}
\label{fig:mermin-scenario}
\end{figure}

\subsection{Vertices}
An {\it outcome assignment} on a subset $\Omega\subset E$ is a function $\gamma:\Omega\to \ZZ_2$ satisfying
$$
\gamma(a)+\gamma(b) = \beta(a,b) +\gamma(a+b)
$$
for all $a,b\in \Omega$ with $\omega(a,b)=0$.  
A subset $\Omega\subset E$ is called {\it closed non-contextual (cnc)} if
\begin{itemize}
\item (closed) $a+b\in \Omega$ for all $a,b\in \Omega$ with $\omega(a,b)=0$, and

\item (non-contextual) there exists an outcome assignment $\gamma:\Omega\to \ZZ_2$.
\end{itemize}
Under the action of ${G}$ there are two orbits of {maximal} cnc sets which we call type $1$ and type $2$ denoted by $\bar{T}_1$ and $\bar{T}_2$. {They are described as follows:
\begin{itemize}
\item A $\bar T_1$ cnc set $\Omega$ consists of $\set{0,a,b,c}$ where $a,b,c\in E^\nloc$ are pairwise anti-commuting elements. %\comm{define anticommuting/commuting}
\item A $\bar T_2$ cnc set $\Omega$ consists of the union of two distinct maximal isotropic subspaces whose non-zero elements belong to $E^\nloc$.
\end{itemize}
There are  $6$
%($6$) 
$\bar{T}_{1}$ and $9$ 
%($9$) 
$\bar{T}_{2}$ cnc sets.  
}

%\comm{we should say maximal cnc when talking about vertices.}

\Thm{[\!\cite{okay2022mermin}]
Vertices of $\MP$ are of the form
$$
v_{(\Omega,\gamma)} = \frac{1}{4} \sum_{a\in \Omega} (-1)^{\gamma(a)} T_a
$$
where $\Omega\subset E$ is a maximal cnc set and $\gamma:\Omega\to \ZZ_2$ is an outcome assignment.
The group ${G}$ acts on the set of vertices, and there are two orbits under this action {corresponding to the two types $\bar T_1$ and $\bar T_2$ of cnc sets}. 
}

%The cnc sets for $\bar{T}_{1}$ vertices consist of the zero element together with three elements $a,b,c \in E^{(\text{nl})}$ which all pairwise anti-commute. 
%The $\bar{T}_{2}$ cnc sets, on the other hand, are given by $\Omega = I\cup I^{\prime}$, where $I,I^{\prime}\subset E^{(\text{nl})}$ {are maximal isot} such that there is a single non-identity element in $I\cap I^{\prime}$. 
%Altogether there are six ($6$) $\bar{T}_{1}$ and nine ($9$) $\bar{T}_{2}$ cnc sets. 

  \begin{figure}[h!]
\centering
\begin{subfigure}{.49\textwidth}
\centering
   \includegraphics[width=.4\linewidth]{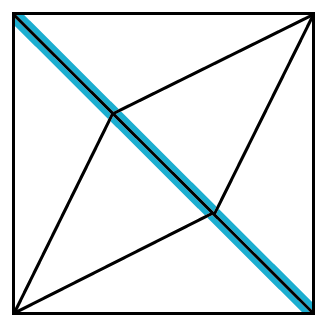}
  \caption{}
  \label{fig:T1}
\end{subfigure}%
\begin{subfigure}{.49\textwidth}
  \centering
   \includegraphics[width=.4\linewidth]{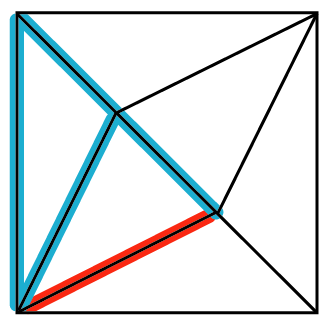}
  \caption{}
  \label{fig:T2}
\end{subfigure}
\caption{(a) A $\bar T_1$ cnc set with an outcome assignment. (b) A $\bar T_2$ cnc set with an outcome assignment. Blue indicates $\gamma(a)=0$ and red indicates $\gamma(a)=1$. 
}
\label{fig:T1-T2}
\end{figure}

The number of outcome assignments is determined by the number of {linearly} independent 
%Pauli 
elements in $\Omega$, which is given by the rank of the parity check matrix whose rows are $a\in \Omega$. For both cnc sets the number of outcome assignments is $8 = 2^{3}$, thus we have:
\begin{itemize}
\item {$\bar T_1$ vertices:}
%Type $1$: 
$48 = 6 \times 8$ pairs $(\Omega,\gamma)$. A representative vertex is 
\begin{equation}\label{eq:T1-representative}
  \centering
  \begin{tabular}{ c|c c c } 
 $1$ & ~~ & ~~ & ~~ \\
 \hline 
 ~~ & $1$ & $1$ & $1$ \\ 
 ~~ & $0$ & $0$ & $0$  \\
 ~~ & $0$ & $0$ & $0$
\end{tabular}
\end{equation}
\item {$\bar T_2$ vertices:}
%Type $2$: 
$72 = 6 \times 8$ pairs $(\Omega,\gamma)$. A representative vertex is 
\begin{equation}\label{eq:T2-representative}
 \centering
  \begin{tabular}{ c|c c c } 
 $1$ & ~~ & ~~ & ~~ \\
 \hline 
 ~~ & $1$ & $1$ & $0$ \\ 
 ~~ & $1$ & $-1$ & $0$  \\
 ~~ & $0$ & $0$ & $1$
\end{tabular}
\end{equation}
\end{itemize}
%\comm{add topological representation of cnc's}
{The topological representations of the vertices  in Eq.~(\ref{eq:T1-representative}) and  Eq.~(\ref{eq:T2-representative}) are given in Fig.~(\ref{fig:T1}) and 
 Fig.~(\ref{fig:T2}), respectively.}

\section{Classical part of the Mermin polytope}\label{sec:classical-mermin}

The DD algorithm is a so-called incremental method for solving the vertex enumeration problem that is dual to the method of Fourier-Motzkin elimination; see e.g., \cite{ziegler2012lectures}. (An overview of the DD algorithm can be found in Appendix~\ref{sec:DDM-VEP}; see also \cite{fukuda2005double}.) The basic idea is that, given the $H$-description of a polyhedron $P$ we iteratively construct its $V$-description by inserting the facet-defining inequalities of $P$ one-by-one.

A typical application of the DD method is when a subset of the inequalities of $P$ define a polyhedron $P^{\prime}$ (where $P \subset P^{\prime}$) for which the vertex enumeration problem has already been solved. When a new inequality $h\cdot x \geq 0$ is introduced, if an extreme point $v\in P^{\prime}$ violates this inequality, i.e. $h\cdot v < 0$, then for all other extreme points $v^{\prime} \in P^{\prime}$ where $h\cdot v > 0$ we construct a new point $u = pv + (1-p)v^{\prime}$ as a convex (or positive) combination such that $h\cdot u = 0$ and then we discard the vertex $v$. In fact, as shown in \cite{fukuda2005double}, it suffices to consider only those $v^{\prime} \in P^{\prime}$ that are {adjacent} to $v$ in the sense that there is an {edge} (or one-dimensional face) in $P^{\prime}$ that connects them. The geometric interpretation of the DD algorithm is highlighted in Fig.~{(\ref{fig:rank-inc})}. Observe, in particular, that the successive introduction of new inequalities promotes non-neighbors (i.e. those vertices not connected by an edge) to neighbors from which new vertices can be produced.

In our implementation of the DD algorithm we consider a target polytope $P_{12} = P_{1}\cap P_{2}$. The {set of} facet-defining inequalities of $P_{12}$ (see Appendix~\ref{sec:intersection}) in this case is just the union of the inequalities of $P_{1}$ and $P_{2}$. We assume that the vertex enumeration problem has been solved for $P_{1}$ and that the facets of $P_{2}$ are known.

\Pro{\label{pro:DD for intersection of pair of polytopes}
Let $P_1,P_2\subset \RR^d$ be a pair of full-dimensional polytopes satisfying the following  conditions:
\begin{enumerate}

\item There are subsets $V_1'\subset V_1$, $H_2'\subset H_2$ and a bijection $\phi:H_2'\to V'_1$. {Moreover, for $h\in H_2'$ we have} 
%such that 
$h\cdot v <0$ 
{where}
%for 
%$h\in H_2'$ and 
$v\in V_1$ if and only if $v=\phi(h)$. 

\item For $h\in H_2$ and $v\in V_1-V_1'$ we have $h\cdot v\geq 0$.  

\item For $v,w\in V_1'$ {that are not neighbors} we have
$$
(h_v\cdot w)(h_w\cdot v) \geq (h_v\cdot v)(h_w\cdot w)
$$ 
where $h_v\in H_2'$ such that $\phi(h_v)=v$, similarly for $h_w$.
\end{enumerate}
Then the vertex 
$$
u = p v + (1-p)w,
$$
where $p\in [0,1]$, 
$v\in V_1'$ and $w\in V_1$, 
of the polytope obtained at a finite stage of the DD algorithm by  insertion of the inequality $h_v$, is a vertex of the intersection polytope $P_{12}=P_1\cap P_2$.
}
{
\Proof{Proof is given in Section \ref{sec:intersection}.}
}

We will apply this result to $(P_1,P_2)=(\MP,{\overline \CL})$. The relevant subsets of inequalities and vertices will be as follows:
\begin{itemize}
\item $V_1'$ is the set of $\bar T_2$ vertices  and $V_1-V_1'$ is the set of $\bar T_1$ vertices of {$\MP$.}
%$\overline{\MP}_1$. 
\item $H_2'$ is the set of CHSH inequalities defining ${\overline \CL}$. Then $H_2-H_2'$ is the set of NN inequalities.
\end{itemize}

\begin{figure}[h!]
\centering
\begin{subfigure}{.245\textwidth}
  \centering
  \includegraphics[width=.8\linewidth]{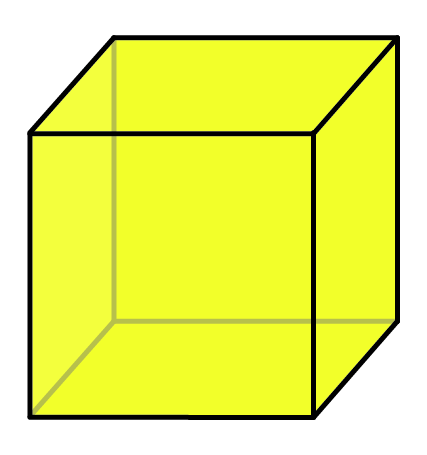}
  \caption{}
  \label{fig:rank-inc-a}
\end{subfigure}%
\begin{subfigure}{.245\textwidth}
  \centering
  \includegraphics[width=.8\linewidth]{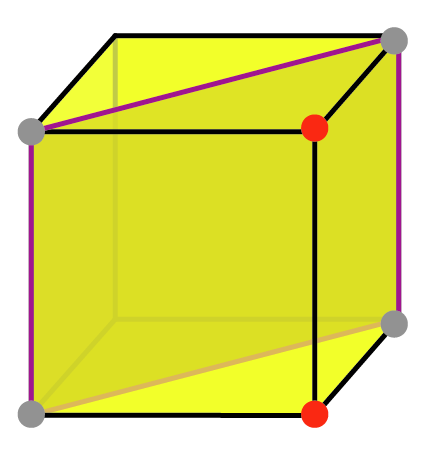}
  \caption{}
  \label{fig:rank-inc-a}
\end{subfigure}%
\begin{subfigure}{.245\textwidth}
  \centering
  \includegraphics[width=.8\linewidth]{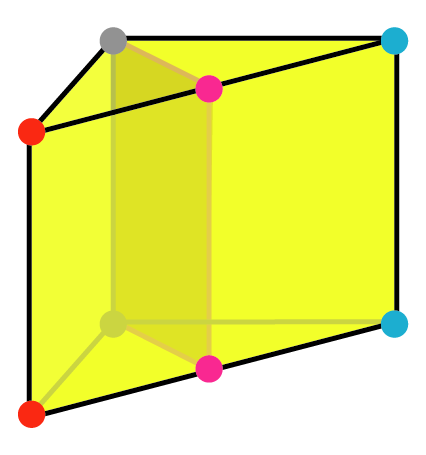}
  \caption{}
  \label{fig:rank-inc-a}
\end{subfigure}%
\begin{subfigure}{.245\textwidth}
  \centering
  \includegraphics[width=.8\linewidth]{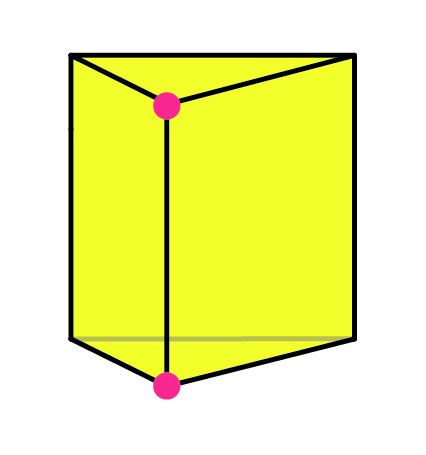}
  \caption{}
  \label{fig:rank-inc-d}
\end{subfigure}%
\caption{{Application of the DD algorithm to the cube.} (a) Our initial polytope, the cube. (b) We cut the cube by a facet. Two vertices (red) violate this facet and four vertices (gray) are tight at this facet. No new vertices are generated, only the two vertices (red) are discarded. However, the joint rank of non-neighbor vertices increases. (c) We cut by an additional facet. The violating vertices (red) are neighbors to vertices (blue) on the other side of the plane. (d) The violating (red) vertices are discarded and two new vertices (pink) are generated.
}
\label{fig:rank-inc}
\end{figure}

\subsection{$\mathbf{\bar{T}_{1}}$ vertices} 
Let us first characterize the CHSH inequalities. Recall that a $\bar T_2$ cnc set $\Omega = I\cup I^{\prime}$ is such that $I\cap I^{\prime} = \{0,a\}$, for some {$a\in E^{(\text{nl})}$}. 
%\comm{note that $E^\nloc$ does not include $0$}
 Let us then attribute to each $\bar{T}_{2}$ cnc set the corresponding set 
\begin{equation}\label{eq:boundary}
\partial\Omega = \Omega - \set{0,a}
\end{equation}
%$\bar{\Omega} = \Omega - I\cap I^{\prime}$ \comm{the $\partial\Omega$ notation for the CHSH is cleaner, imo. The bar notation suggest that $\Omega$ and $\bar \Omega$ are the same kind of objects. I prefer we keep some geometric intuition that is helpful for combinatorial arguments.}
together with a function $\gamma:\partial\Omega\to \ZZ_{2}$ such that $\sum_{a\in\partial\Omega} \gamma(a) = 1~\text{mod}~2$. {The elements of $\partial\Omega$ corresponds to the elements on the boundary of a $\bar T_2$ cnc set; see Fig.~(\ref{fig:T2}).}
%\comm{explain the boundary interpretation}
The CHSH inequalities take the form%
\begin{eqnarray}
h_{\left (\partial\Omega,{\gamma}\right )}\cdot {x}\geq 0,%
\label{eq:chsh-inequality}
\end{eqnarray}
where $x\in \RR^{10}$ ($x_{0}=1$) and where $h=h_{\left (\partial\Omega,{\gamma}\right )}\in \RR^{10}$ is defined by
\begin{eqnarray}
h_a = \left\lbrace
\begin{array}{cc}
2 & a = 0 \\
(-1)^{{\gamma}(a)} & a\in \partial\Omega \\
0 & \text{otherwise.}
\end{array} 
 \right.
\end{eqnarray} 
% One consequence of this result is that convex decompositions are preserved under the induced action of $G_{1}$ on $\overline{\MP}$. \comm{where do we use symmetry argument?}
Let $v$ be a vertex corresponding to $(\Omega,\gamma)$ and $h$ be a CHSH inequality corresponding to $(\partial\Omega',\gamma')$. Then 
 \begin{equation}\label{eq:h dot v}
h\cdot v = 2+\sum_{a\in{\Omega \cap \partial\Omega'}}(-1)^{\gamma(a)+{\gamma}^{\prime}(a)}.
\end{equation}

  \begin{figure}[h!]
\centering
\begin{subfigure}{.49\textwidth}
  \centering
   \includegraphics[width=.8\linewidth]{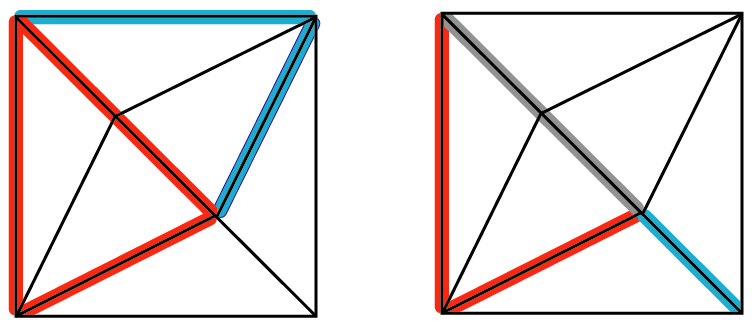}
  \caption{}
  \label{fig:T1 intersection}
\end{subfigure}%
\begin{subfigure}{.49\textwidth}
  \centering
   \includegraphics[width=.8\linewidth]{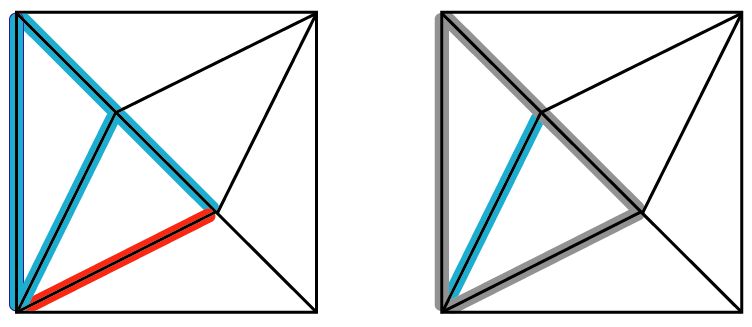}
  \caption{}
  \label{fig:T2 intersection}
\end{subfigure}
\caption{Let $\Omega'$ be a $\bar T_2$ cnc set. (a) Intersection of $\bar T_1$ cnc   with the boundary $\partial \Omega'$. The sizes of the intersections are $0$ and $2$.  
(b) Intersection with a $\bar T_2$ cnc  set. The sizes of the intersections are $2$ and $4$.   
}
\label{fig:T1-T2 intersection}
\end{figure}

{
\Lem{\label{lem:intersection of cnc with boundary cnc}
Let $\Omega$ be a $\bar T_i$ cnc set and $\Omega'$ be a $\bar T_2$ cnc set. Then  
$$
|\Omega\cap \partial \Omega'| =
\left\lbrace
\begin{array}{ll}
0,2 & i=1\\
2,4 & i=2.
\end{array}
\right.
$$
}
\Proof{
We depict elements of a cnc set $\Omega$ (blue), $\partial \Omega^{\prime}$ (red),  and their intersection (gray):
\begin{itemize}
\item For $\bar T_{1}$ the intersection is either empty or contains $2$ elements:
$$
\begin{tabular}{ c|c c c } 
 {\Large $\color{white}{ \bullet}$} & ~~ & ~~ & ~~ \\
 \hline 
 ~ & {\Large $\color{red}{ \bullet}$} & {\Large $\color{red}{ \bullet}$} & {\Large $\color{white}{ \bullet}$} \\ 
 ~ & {\Large $\color{red}{ \bullet}$} & {\Large $\color{red}{ \bullet}$} &  {\Large $\color{white}{ \bullet}$} \\
 ~ & {\Large $\color{cyan-custom}{ \bullet}$} & {\Large $\color{cyan-custom}{ \bullet}$} & {\Large $\color{cyan-custom}{ \bullet}$}
\end{tabular}
\hspace{1 cm}%
\begin{tabular}{ c|c c c } 
 {\Large $\color{white}{ \bullet}$} & ~~ & ~~ & ~~ \\
 \hline 
 ~ & {\Large $\color{gray-custom}{ \bullet}$} & {\Large $\color{gray-custom}{ \bullet}$} & {\Large $\color{cyan-custom}{ \bullet}$} \\ 
 ~ & {\Large $\color{red}{ \bullet}$} & {\Large $\color{red}{ \bullet}$} &  {\Large $\color{white}{ \bullet}$} \\
 ~ & {\Large $\color{white}{ \bullet}$} & {\Large $\color{white}{ \bullet}$} & {\Large $\color{white}{ \bullet}$}
\end{tabular}
$$
\item For $\bar T_{2}$ 
{either the intersection contains $2$ elements or is $\partial \Omega$ itself}:
$$
\begin{tabular}{ c|c c c } 
 {\Large $\color{white}{ \bullet}$} & ~~ & ~~ & ~~ \\
 \hline 
 ~ & {\Large $\color{red}{ \bullet}$} & {\Large $\color{gray-custom}{ \bullet}$} & {\Large $\color{cyan-custom}{ \bullet}$} \\ 
 ~ & {\Large $\color{red}{ \bullet}$} & {\Large $\color{gray-custom}{ \bullet}$} &  {\Large $\color{cyan-custom}{ \bullet}$} \\
 ~ & {\Large $\color{cyan-custom}{ \bullet}$} & {\Large $\color{white}{ \bullet}$} & {\Large $\color{white}{ \bullet}$}
\end{tabular}\hspace{1 cm}%
\begin{tabular}{ c|c c c } 
 {\Large $\color{white}{ \bullet}$} & ~~ & ~~ & ~~ \\
 \hline 
 ~ & {\Large $\color{gray-custom}{ \bullet}$} & {\Large $\color{gray-custom}{ \bullet}$} & {\Large $\color{white}{ \bullet}$} \\ 
 ~ & {\Large $\color{gray-custom}{ \bullet}$} & {\Large $\color{gray-custom}{ \bullet}$} &  {\Large $\color{white}{ \bullet}$} \\
 ~ & {\Large $\color{white}{ \bullet}$} & {\Large $\color{white}{ \bullet}$} & {\Large $\color{cyan-custom}{ \bullet}$}
\end{tabular}
$$
 \end{itemize}
%\caption{We depict elements of a cnc set $\Omega$ (blue), $\bar \Omega^{\prime}$ (red),  and their intersection (gray). (a) For $\bar T_{1}$ the intersection is either empty of contains $2$ elements. (b) For $\bar T_{2}$ the intersection is either $\bar \Omega$ itself or contains $2$ elements.
See Fig.~(\ref{fig:T1-T2 intersection}) for the corresponding topological representations. We observe that these are the only possible sizes of the intersection. 
%\comm{add figure.}
}
}

\begin{pro}\label{pro:type1-classical}
All $\bar{T}_{1}$ vertices of $\MP$ fall within {$\overline\CL$}. 
This implies that all $\bar{T}_{1}$ vertices are also vertices of $\overline{\MP}$.
\end{pro}
\begin{proof}
{Let $\Omega=\set{0,a,b,c}$ be a $\bar T_1$ cnc set.
By Lemma \ref{lem:intersection of cnc with boundary cnc} and Eq.~(\ref{eq:h dot v}) we have 
\begin{equation}\label{eq:h dot v for T1}
h\cdot v = \left\lbrace
\begin{array}{ll}
2 & \Omega\cap\partial\Omega'= \emptyset \\
2+ (-1)^{\gamma(a)+\gamma'(a)}+(-1)^{\gamma(b)+\gamma'(b)} & \Omega\cap\partial\Omega'=\set{a,b}.
\end{array}
\right.
\end{equation}
Therefore $h\cdot v =0,2,4$.
}   
\end{proof}

\subsection{$\mathbf{\bar{T}_{2}}$ vertices}
 
We will write $\bar \gamma$ for the outcome assignment defined by 
$$
\bar\gamma(a') = \left\lbrace
\begin{array}{ll}
 \gamma(a')+1 & a'\in \partial\Omega \\
 \gamma(a') & \text{otherwise}.
\end{array}
\right.
$$ 
Then we can define a function
$$
\phi:H_2'\to V'_1
$$
by sending a CHSH inequality $h_{(\partial \Omega,\gamma)}$ to the $\bar T_2$ vertex $v_{(\Omega,\bar \gamma)}$.  
The pair $(h,\phi(h))$ is called {\it dual}.

\begin{lem}\label{lem:type2-violate} 
Let $v$ be a $\bar{T}_{2}$ vertex represented by $(\Omega,\gamma)$ and $h$ be a CHSH inequality represented by $(\partial\Omega^{\prime},{\gamma}^{\prime})$. Then
\begin{enumerate}
\item $h\cdot v=-2$ if $\Omega\cap \partial\Omega^{\prime} = \partial\Omega^{\prime}$ and  $\gamma(a)={\gamma}^{\prime}(a)+1$ for all $a\in \partial\Omega'$,
\item $h\cdot v=0$ if $\Omega \cap\partial\Omega'=\set{a,b}$, $\gamma(a)=\gamma'(a)+1$, and $\gamma(b)=\gamma'(b)+1$,
\item otherwise $h\cdot v\in \set{2,4,6}$. 
\end{enumerate} 
\end{lem}
\begin{proof}  
{
This follows from Lemma \ref{lem:intersection of cnc with boundary cnc} and Eq.~(\ref{eq:h dot v}):
\begin{equation}\label{eq:h dot v for T2}
h\cdot v = \left\lbrace
\begin{array}{ll}
-2 & \Omega \cap\partial\Omega'= \partial\Omega,\; \gamma'=\bar\gamma\\
0 & \Omega \cap\partial\Omega'=\set{a,b},\; \gamma(a)=\gamma'(a)+1,\;\gamma(b)=\gamma'(b)+1\\
2,4,6 &\text{otherwise.}   
\end{array}
\right.
\end{equation}
}
\end{proof}

\begin{lem}\label{lem:extended-function}
Consider a pair $(\partial\Omega,\gamma)$ representing a CHSH inequality. The function $\gamma$ can be uniquely extended to an outcome assignment $\tilde \gamma$ on the whole set $\Omega$.
\end{lem}
\begin{proof} 
{The set $\partial\Omega$ contains four elements $b,c,b',c'$ and $a\in \Omega-\partial\Omega$ can be written as $a=b+c=b'+c'$.}
%Recall that for $\Omega = I\cup I^{\prime}$ such that $I\cap I^{\prime} = \{0,a\}$ we have that $\bar \Omega := \Omega - \{0,a\}$. Consider non-identity elements $b,c \in I$ and $b^{\prime},c^{\prime}\in I^{\prime}$ such that $a = b+c = b^{\prime}+c^{\prime}$. 
To define the extended function $\tilde \gamma$ set $\tilde \gamma(0) = 0$ and 
%note that 
%$\tilde \gamma(a)$ is uniquely determined by both 
$\tilde \gamma(a) := \gamma(b)+\gamma(c)+\beta(b,c)$.
{Then this will agree with}
 $\tilde \gamma(a) := \gamma(b^{\prime})+\gamma(c^{\prime})+\beta(b^{\prime},c^{\prime})$ {since $\sum_{a\in \partial\Omega}  \gamma(a)=1~\text{mod}~2$ and $\beta(b,c)+\beta(b',c')=1~\text{mod}~2$.} 
% . Let us check that this is always consistent. Set the two expressions equal to one another and use the fact that if $\beta(b,c) = r \in \ZZ_{2}$, then $\beta(b^{\prime},c^{\prime}) = r+1 \in \ZZ_{2}$. We then have that $\sum_{e\in \bar \Omega}\gamma(e) = 1~\text{mod}~2$, which is always satisfied by $\gamma$.
\end{proof}

\Pro{\label{pro:condition-1}
Let $H_2'\subset H_2$ denote the subset consisting of the CHSH inequalities.
Then $h\cdot v\geq 0$ for every $h\in H_2-H_2'$ and $v\in V_1$. 
There is a bijection  
$$
\phi:H_2' \to {V_1'}
$$
that sends a CHSH inequality to a $\bar T_2$ vertex such that $h\cdot v<0$ for $h\in H_2'$ and $v\in V_1$ if and only if $v=\phi(h)$.
} 
\Proof{ 
The set $H_2-H_2'$ consists of the 
%non-negativity 
{NN}
inequalities, which are satisfied by the points in {$\MP$} since they are already in the range $[-1,1]$.
{The bijection follows from the unique extension constructed in Lemma \ref{lem:extended-function}. The statement about $h\cdot v$ follows from part (1) of Lemma \ref{lem:type2-violate}.} 
}

{
\begin{lem}\label{lem:tight-at-chsh}
There are  $8$ $\bar T_1$ and $16$ $\bar T_{2}$ vertices that are tight at each CHSH inequality.
\end{lem}
\begin{proof}
%\comm{revise}
%\si{%
{From Eq.~(\ref{eq:h dot v for T1}) and part (2) of Lemma \ref{lem:type2-violate} we observe that $h\cdot v=0$ when $\Omega\cap \partial\Omega'=\set{a,b}$ with $\gamma(a)=\gamma'(a)+1$ and $\gamma(b)=\gamma'(b)+1$.
%Let $\Omega$ be a $\bar T_1$ cnc set.  Eq.~(\ref{eq:h dot v for T1}) implies that $h\cdot v=0$ when $\Omega\cap \partial\Omega'=\set{a,b}$ with $\gamma(a)=\gamma'(a)+1$ and $\gamma(b)=\gamma'(b)+1$. 
For a fixed $\Omega'$ there are $4$ $\bar T_1$ {type} $\Omega$'s and for each such set there are $2$ $\gamma$'s.
Now consider the case where $\Omega$ is a $\bar T_2$ cnc set.
%From Lemma \ref{lem:type2-violate} part(2) we see that $h\cdot v=0$ when $\Omega\cap\partial\Omega'=\set{a,b}$, $\gamma(a)=\gamma'(a)+1$ and $\gamma(b)=\gamma'(b)+1$. 
There are $8$ such $\Omega$'s and $2$ $\gamma$'s for each.
% such cnc set.
}
%The outcome assignment $\gamma$ is fixed on $\Omega\cap $
%This proof follows analogously to that of Lemma~\ref{lem:tight-chsh} except now we fix the CHSH pair $(\partial\Omega,\gamma)$ and vary over cnc vertices {parametrized by} $(\Omega^{\prime},\gamma^{\prime})$.
%}
\end{proof}
}

%\noindent \si{[[This Lemma can probably be combined with Lemma~\ref{lem:chsh-values}.]}\\

\begin{pro}\label{pro:neighbors}
%Let $\bar{v}$ be a $T_{2}$ vertex indexed by $(\Omega_{2},s)$ with neighbors in $\MP$ parameterized by the set $N_{v}$ and let $\bar{B}\bar{x}\geq 0$ be the inequality \textit{dual} to $v$. For all $j\in N_{i}$ we have that $\bar{B}\bar{v}_{j} = 0$. 
Let $h$ be a CHSH inequality and $v=\phi(h)$ be the dual $\bar T_2$ vertex. We have that $h\cdot w=0$ a 
%$T_2$ 
vertex $w$ if and only if 
%$w=\phi(v)$ 
{$w$ is a neighbor of $v$}. 
%\si{[[@Cihan: What is this lemma saying? I am not sure it is correct.]]}
\end{pro}
\begin{proof}
%\comm{revise}
The local structure of $\MP$ was characterized in \cite{okay2022mermin} and summarized in Section~\ref{sec:comb-mp1}. 
{A $\bar T_2$ vertex has $8$ $\bar T_1$ vertices and $16$ $\bar T_2$ vertices.} 
{By Corollary \ref{cor:intersection two elements and outcome assignments coincide} and Eq.~(\ref{eq:h dot v}) we observe that $h\cdot w=0$ whenever $w$ is a neighbor of $v$. Lemma~\ref{lem:tight-at-chsh} implies that these are all the vertices tight at $h$. }
%Each $\bar T_{2}$ vertex has $8$ $\bar T_{1}$ neighbors and $16$ $\bar T_{2}$ neighbors, {whereas a $\bar T_1$ vertex has $12$ $\bar T_2$ neighbors.} \comm{numbers match up but we still need the structure of neighbors to conclude that they are tight.} Comparison with Lemma~\ref{lem:tight-at-chsh} {that all of these neighbors are tight at $h$.} 
%yields the result.
%Converse{ly} 
%follows from  
%Lemmas~\ref{lem:tight-chsh} and \ref{lem:tight-at-chsh}
%{implies that these are the only vertices tight at $h$.}.
\end{proof}

%\si{
%\begin{rem}
%{\rm
See Fig.~(\ref{fig:duality}) for an illustration of the duality between $h$ and $\phi(h)$ as well as the corresponding neighborhood structure.
%}
%\end{rem}
% }

\subsection{Vertex enumeration} 
Proposition \ref{pro:DD for intersection of pair of polytopes} is our main tool to generate the vertices of $\overline{\MP}$.
To apply this result we first check that the Conditions (1)-(3) are satisfied by our pair $(P_1,P_2)=({\MP},{\overline \CL})$:
\begin{itemize}
\item Condition (1) is proved in {Proposition~\ref{pro:condition-1}.}
\item Condition (2) follows from Proposition \ref{pro:type1-classical}.
\item Condition (3) can be proved as follows: By part (1) of Lemma \ref{lem:type2-violate} we have $(h_v\cdot v)(h_w\cdot w)\leq 4$. By parts (2) and (3) of the same lemma and Proposition \ref{pro:neighbors} we have $(h_v\cdot w)(h_w\cdot v) \geq 4$.
\end{itemize}
For our starting point we consider the canonical $\bar T_2$ vertex (see Fig.~\ref{fig:T2}):
%\left (I + XX + XY + YX - YY + ZZ \right ).
\begin{eqnarray}
{v =  
 \centering
  \begin{tabular}{ c|c c c } 
 $1$ & ~~ & ~~ & ~~ \\
 \hline 
 ~~ & $1$ & $1$ & $0$ \\ 
 ~~ & $1$ & $-1$ & $0$  \\
 ~~ & $0$ & $0$ & $1$
\end{tabular}
}
\label{eq:canonical-vertex}
\end{eqnarray}
The DD algorithm will run by inserting a sequence of inequalities $h_1,\cdots,h_k \in H_2$ such that $h_k=h_v$, and at {the} final step we will obtain 
$$
u=pv+(1-p)w
$$ 
where $w$ is a vertex of $\MP$ and
$$
p = \frac{1}{1-\frac{hv\cdot v}{h_v\cdot w}}\;\in [0,1].
$$
Then Proposition \ref{pro:DD for intersection of pair of polytopes} will imply that $u$ is a vertex of $\overline{\MP}$. 
{By Proposition \ref{pro:neighbors} the substitution of a CHSH inequality in the DD algorithm can increase the joint rank $\rank(v,w)$ of a pair of vertices $v,w\in V_1$ if and only if $v$ and $w$ are not neighbors, and they are  neighbors of a vertex $v'$. Therefore we will consider a pair $v,w$ of non-neighbors with $v'$ as their common neighbor.}
%for a pair $v,w\in V_1$ the only way to increase the joint rank $\rank(v,w)$ by the substitution of a CHSH inequality  in the DD algorithm is to choose {two non-neighbor vertices} $v$ and $w$ {that are common neighbors of a vertex $v'$.}
%$2$-neighbors \comm{ref to App} \si{[[@Cihan: Some of the $w$s are $3$-neighbors. Is this a different definition?]]}.
Consider the canonical $\bar T_2$ neighbor of $v$:
%v' = \frac{1}{4}\left (I + XX - YY - YZ - ZY + ZZ \right ).%
\begin{eqnarray}
{v' =  
 \centering
  \begin{tabular}{ c|c c c } 
 $1$ & ~~ & ~~ & ~~ \\
 \hline 
 ~~ & $1$ & $0$ & $0$ \\ 
 ~~ & $0$ & $-1$ & $-1$  \\
 ~~ & $0$ & $-1$ & $1$
\end{tabular}
 }
\label{eq:canonical-vertex}
\end{eqnarray}

%We will use the following common neighbors of $v$ and $v'$:

\begin{pro}\label{pro:wi}
%The neighbors of $v'$ that are not also neighbors of $v$ break into four orbits under the action of the stabilizer of $v$ with the following representative vertices:
The following vertices of $\MP$
%w_{0} &=& \frac{1}{4}\left (I + XX - XY - YZ - ZX - ZY \right ),\label{eq:u0}\\
%w_{1} &=& \frac{1}{4}\left (I + XX - XY - YX - YY + ZZ \right ),\label{eq:u1}\\
%w_{2} &=& \frac{1}{4}\left (I - XY - YY - ZY \right ),\label{eq:u2}\\
%w_{3} &=& \frac{1}{4}\left (I - XY - YX - YZ - ZX + ZZ \right ).\label{eq:u3}
\begin{subequations}
\begin{eqnarray}
w_0 &=&  
 {\centering
  \begin{tabular}{ c|c c c } 
 $1$ & ~~ & ~~ & ~~ \\
 \hline 
 ~~ & $1$ & $-1$ & $0$ \\ 
 ~~ & $0$ & $0$ & $-1$  \\
 ~~ & $-1$ & $-1$ & $0$
\end{tabular}}\label{eq:u0}\\
w_1 &=&  
 {\centering
  \begin{tabular}{ c|c c c } 
 $1$ & ~~ & ~~ & ~~ \\
 \hline 
 ~~ & $1$ & $-1$ & $0$ \\ 
 ~~ & $0$ & $0$ & $-1$  \\
 ~~ & $-1$ & $-1$ & $0$
\end{tabular}}\label{eq:u1}\\
w_2 &=&  
 {\centering
  \begin{tabular}{ c|c c c } 
 $1$ & ~~ & ~~ & ~~ \\
 \hline 
 ~~ & $0$ & $-1$ & $0$ \\ 
 ~~ & $0$ & $-1$ & $0$  \\
 ~~ & $0$ & $-1$ & $0$
\end{tabular}}\label{eq:u2}\\
w_3 &=&  
 {\centering
  \begin{tabular}{ c|c c c } 
 $1$ & ~~ & ~~ & ~~ \\
 \hline 
 ~~ & $0$ & $-1$ & $0$ \\ 
 ~~ & $-1$ & $0$ & $-1$  \\
 ~~ & $-1$ & $0$ & $1$
\end{tabular}}\label{eq:u3}
\end{eqnarray}
\end{subequations}
are neighbors of $v'$ but not neighbors of $v$, and the joint rank {(with respect to the generating matrix $\bar M^\nloc$ of $\MP$)} is given by 
$$
\rank(v,w_i) = \left\lbrace
\begin{array}{ll}
 7  & i=0 \\ 
 6  & \text{otherwise.}
 \end{array}
\right.
$$
%of each $u_{i}$ with $\bar{v}_{0}$ (in order) is $\text{rank}_{\bar{A}}(\bar{v}_{0},\bar{u}_{i}) = 6,7,6,6$.
\end{pro}
\Proof{{See Corollary \ref{cor:2 neighbors}.} 
}

Next we identify the CHSH inequalities that are tight at the pairs $(v,w_i)$.
They will be used to increase the joint rank.

\Lem{\label{lem:tight CHSH}
The CHSH inequalities tight at   both $v$ and $w_i$ are given as follows:
\begin{enumerate}
\item $i=0$
$$
  \centering
  \begin{tabular}{ c|c c c } 
 $2$ & ~~ & ~~ & ~~ \\
 \hline 
 ~~ & $0$ & $0$ & $0$ \\ 
 ~~ & $0$ & $1$ & $1$  \\
 ~~ & $0$ & $1$ & $-1$
\end{tabular}\hspace{0.5 cm}%
\begin{tabular}{ c|c c c } 
 $2$ & ~~ & ~~ & ~~ \\
 \hline 
 ~~ & $0$ & $0$ & $0$ \\ 
 ~~ & $-1$ & $1$ & $0$  \\
 ~~ & $1$ & $1$ & $0$
\end{tabular}\hspace{0.85 cm}%
\begin{tabular}{ c|c c c } 
 $2$ & ~~ & ~~ & ~~ \\
 \hline 
 ~~ & $-1$ & $0$ & $-1$ \\ 
 ~~ & $0$ & $0$ & $0$  \\
 ~~ & $1$ & $0$ & $-1$
\end{tabular}\hspace{0.5 cm}%
\begin{tabular}{ c|c c c } 
 $2$ & ~~ & ~~ & ~~ \\
 \hline 
 ~~ & $-1$ & $0$ & $-1$ \\ 
 ~~ & $-1$ & $0$ & $1$  \\
 ~~ & $0$ & $0$ & $0$
\end{tabular}
$$

\item $n=1$
$$
  \centering
  \begin{tabular}{ c|c c c } 
 $2$ & ~~ & ~~ & ~~ \\
 \hline 
 ~~ & $0$ & $0$ & $0$ \\ 
 ~~ & $0$ & $1$ & $1$  \\
 ~~ & $0$ & $1$ & $-1$
\end{tabular}\hspace{0.5 cm}%
\begin{tabular}{ c|c c c } 
 $2$ & ~~ & ~~ & ~~ \\
 \hline 
 ~~ & $0$ & $0$ & $0$ \\ 
 ~~ & $0$ & $1$ & $-1$  \\
 ~~ & $0$ & $-1$ & $-1$
\end{tabular}\hspace{0.85 cm}%
\begin{tabular}{ c|c c c } 
 $2$ & ~~ & ~~ & ~~ \\
 \hline 
 ~~ & $-1$ & $0$ & $1$ \\ 
 ~~ & $0$ & $0$ & $0$  \\
 ~~ & $-1$ & $0$ & $-1$
\end{tabular}\hspace{0.5 cm}%
\begin{tabular}{ c|c c c } 
 $2$ & ~~ & ~~ & ~~ \\
 \hline 
 ~~ & $-1$ & $0$ & $-1$ \\ 
 ~~ & $0$ & $0$ & $0$  \\
 ~~ & $1$ & $0$ & $-1$
\end{tabular}
$$
\item $i=2$
$$
  \centering
  \begin{tabular}{ c|c c c } 
 $2$ & ~~ & ~~ & ~~ \\
 \hline 
 ~~ & $0$ & $0$ & $0$ \\ 
 ~~ & $0$ & $1$ & $1$  \\
 ~~ & $0$ & $1$ & $-1$
\end{tabular}\hspace{0.5 cm}
\begin{tabular}{ c|c c c } 
 $2$ & ~~ & ~~ & ~~ \\
 \hline 
 ~~ & $0$ & $0$ & $0$ \\ 
 ~~ & $-1$ & $1$ & $0$  \\
 ~~ & $1$ & $1$ & $0$
\end{tabular}
$$

\item $i=3$
$$
\centering
  \begin{tabular}{ c|c c c } 
 $2$ & ~~ & ~~ & ~~ \\
 \hline 
 ~~ & $0$ & $0$ & $0$ \\ 
 ~~ & $0$ & $1$ & $1$  \\
 ~~ & $0$ & $1$ & $-1$
\end{tabular}\hspace{0.5 cm}
\begin{tabular}{ c|c c c } 
 $2$ & ~~ & ~~ & ~~ \\
 \hline 
 ~~ & $-1$ & $0$ & $-1$ \\ 
 ~~ & $0$ & $0$ & $0$  \\
 ~~ & $1$ & $0$ & $-1$
\end{tabular}
$$
\end{enumerate}
}
\Proof{
%For this we follow an approach  similar to the one used in the proof of Lemma~\ref{lem:tight-chsh}. 
%More explicitly, 
{Let} the pairs $(\Omega_{{v}},\gamma_{v})$ and $({\Omega}_{{w}_{i}},\gamma_{w_{i}})$ correspond to vertices $v$ and $w_{i}$, respectively. We run over all CHSH inequalities {parametrized by $(\partial\Omega,\gamma)$} satisfying
\begin{eqnarray}
|\Omega_{v}\cap \partial\Omega | = 2\quad \text{and}\quad%
\gamma_{v}(a) = \gamma(a)+1\quad \text{for all}\quad \Omega_{v}\cap \partial\Omega,
%s_{\bar{v}_{0}}|_{\Omega^{\perp}_{\bar{v}_{0}}} = \gamma|_{\Omega^{\perp}_{\bar{v}_{0}}} +1,\notag
\notag
\end{eqnarray}
simultaneously with analogous relations holding for $w_{i}$; {see part (2) of Lemma \ref{lem:type2-violate}.} One can directly compute tight CHSH inequalities of each vertex and compare, yielding the mutual tight inequalities.
 
}

\Thm{\label{thm:vertices of MP bar}
The convex combinations 
\begin{equation}\label{eq:ui}
u_i=p v + (p-1) w_i
\end{equation}
where
$$
p=\left\lbrace
\begin{array}{ll}
2/3 &i=3\\
1/2 &\text{otherwise,}
\end{array}
\right.
$$
produce the following vertices of  $\overline{\MP}$:
 $$
\begin{aligned}
u_0&=   \begin{tabular}{ c|c c c } 
 $1$ & ~~ & ~~ & ~~ \\
 \hline 
 ~~ & $1$ & $0$ & $0$ \\ 
 ~~ & $1/2$ & $-1/2$ & $-1/2$  \\
 ~~ & $-1/2$ & $-1/2$ & $1/2$
\end{tabular}%
&  
u_1&=     \begin{tabular}{ c|c c c } 
 $1$ & ~~ & ~~ & ~~ \\
 \hline 
 ~~ & $1$ & $0$ & $0$ \\ 
 ~~ & $0$ & $-1$ & $0$  \\
 ~~ & $0$ & $0$ & $1$
\end{tabular}%
\\
u_2&=     \begin{tabular}{ c|c c c } 
 $1$ & ~~ & ~~ & ~~ \\
 \hline 
 ~~ & $1/2$ & $0$ & $0$ \\ 
 ~~ & $1/2$ & $-1$ & $0$  \\
 ~~ & $0$ & $-1/2$ & $1/2$
\end{tabular}
& 
u_3&=   \begin{tabular}{ c|c c c } 
 $1$ & ~~ & ~~ & ~~ \\
 \hline 
 ~~ & $2/3$ & $1/3$ & $0$ \\ 
 ~~ & $1/3$ & $-2/3$ & $-1/3$  \\
 ~~ & $-1/3$ & $0$ & $1$
\end{tabular}
\end{aligned} 
$$
}  
\Proof{%\comm{revise}
We now compute the change in rank. 
{Let $\bar M$ denote the generating matrix of $\overline\MP$.}
Let 
%$\zZ_{\bar{A}}$ 
{$\zZ_{v}$} 
and 
%$\zZ_{\bar{B}}$ 
{$\zZ_{w_i}$} 
index the joint tight inequalities of {$v$} and some {$w_{i}$}. We are interested in computing the rank of the matrix $\bar{M}[\zZ_{v}\cup \zZ_{w_i}]$.
%, where $\bar{M}$ is the generating matrix for $\overline{\MP}$. 
Notice first that since $\text{rank}(\bar{M}[\zZ_{v}]) = k$, we have that the {(augmented)} matrix $\begin{bmatrix} \bar{M}[\zZ_{v}] & \one_{|\zZ_{v}|\times 1}] \end{bmatrix}$ that includes the constant term can be put in reduced row echelon. In particular, the Pauli expectations in the pivot columns (with leading ones) are dependent on the Pauli expectations in the non-pivot columns. In our case it is straightforward to compute the reduced echelon form algebraically. The result is given in Table~\ref{tab:matrix-tight}.%

\begin{table}[h!]
\centering
\begin{subfigure}{.45\textwidth}
  \centering
  {\footnotesize
  \begin{eqnarray}
%&&\begin{matrix}%
%XX & XY &  XZ & YX & YY & YZ & ZX &  ZY &  ZZ & II
%\end{matrix}\\%
\begin{bmatrix}
 ~~1 & ~~ & ~~ & ~~ & ~~ & ~~ & ~~ & ~~ & ~~ & -1~~ \\
 ~~ & 1 & ~~ & ~~ & ~~ & ~~ & ~~ & -1 & -1 & ~~ \\
 ~~ & ~~ & 1 & ~~ & ~~ & ~~ & ~~ & ~~ & ~~ & ~~ \\
 ~~ & ~~ & ~~ & 1 & ~~ & ~~ & ~~ & -1 & ~~ & -1 \\
 ~~ & ~~ & ~~ & ~~ & 1 & ~~ & ~~ & ~~ & 1 & ~~ \\
 ~~ & ~~ & ~~ & ~~ & ~~ & 1 & ~~ & -1 & ~~ & ~~ \\
 ~~ & ~~ & ~~ & ~~ & ~~ & ~~ & 1 & ~~ & -1 & 1
  \end{bmatrix}\notag
  \end{eqnarray}
  }% close font size
  \caption{}\label{tab:rref-u0}
\end{subfigure}\hspace{1.25 cm}%
\begin{subfigure}{.45\textwidth}
  \centering
  {\footnotesize
  \begin{eqnarray}
 \begin{bmatrix}
%  XX & XY &  XZ & YX & YY & YZ & ZX &  ZY &  ZZ & II\\
 ~~1 & ~~ & ~~ & ~~ & ~~ & ~~ & ~~ & ~~ & ~~ & -1~~ \\ 
 ~~ & 1 & ~~ & -1 & ~~ & ~~ & ~~ & ~~ & ~~ & ~~ \\ 
 ~~ & ~~ & 1 & ~~ & ~~ & ~~ & 1 & ~~ & ~~ & ~~ \\ 
 ~~ & ~~ & ~~ & ~~ & 1 & ~~ & ~~ & ~~ & ~~ & 1 \\ 
 ~~ & ~~ & ~~ & ~~ & ~~ & 1 & ~~ & -1 & ~~ & ~~ \\ 
 ~~ & ~~ & ~~ & ~~ & ~~ & ~~ & ~~ & ~~ & 1 & -1
  \end{bmatrix}\notag
  \end{eqnarray}
  }%close font
  \caption{}\label{tab:rref-u1}
\end{subfigure}%
\vspace{1 em}
\begin{subfigure}{.45\textwidth}
  \centering
  {\footnotesize
  \begin{eqnarray}
 \begin{bmatrix}
%  XX & XY &  XZ & YX & YY & YZ & ZX &  ZY &  ZZ & II\\
 ~~1 & ~~ & ~~ & ~~ & ~~ & ~~ & ~~ & ~~ & -1 & ~~~~ \\ 
 ~~ & 1 & ~~ & ~~ & ~~ & ~~ & 1 & -1 & -1 & ~~ \\ 
 ~~ & ~~ & 1 & ~~ & ~~ & ~~ & 1 & ~~ & ~~ & ~~ \\ 
 ~~ & ~~ & ~~ & 1 & ~~ & ~~ & 1 & -1 & ~~ & -1 \\ 
 ~~ & ~~ & ~~ & ~~ & 1 & ~~ & ~~ & ~~ & ~~ & 1 \\ 
 ~~ & ~~ & ~~ & ~~ & ~~ & 1 & ~~ & -1 & 1 & -1
  \end{bmatrix}\notag
  \end{eqnarray}
  }% close font size
  \caption{}\label{tab:rref-u2}
\end{subfigure}\hspace{1.25 cm}%
\begin{subfigure}{.45\textwidth}
  \centering
  {\footnotesize
  \begin{eqnarray}
 \begin{bmatrix}
%  XX & XY &  XZ & YX & YY & YZ & ZX &  ZY &  ZZ & II\\
 ~~1 & ~~ & ~~ & ~~ & ~~ & -1 & ~~ & 1 & ~~ & -1~~ \\ 
 ~~ & 1 & ~~ & ~~ & ~~ & -1 & -1 & ~~ & ~~ & -1 \\ 
 ~~ & ~~ & 1 & ~~ & ~~ & -1 & 1 & 1 & ~~ & ~~ \\ 
 ~~ & ~~ & ~~ & 1 & ~~ & -1 & -1 & ~~ & ~~ & -1 \\ 
 ~~ & ~~ & ~~ & ~~ & 1 & 1 & ~~ & -1 & ~~ & 1 \\ 
 ~~ & ~~ & ~~ & ~~ & ~~ & ~~ & ~~ & ~~ & 1 & -1 
  \end{bmatrix}\notag
  \end{eqnarray}
  }%close font
  \caption{}\label{tab:rref-u3}
\end{subfigure}%
\caption{We label the columns of $\begin{bmatrix} {\bar M}[\zZ_{v}] & I_{|\zZ_{v}|\times 1} \end{bmatrix}$ by Pauli coefficients according to {$XX,XY,XZ,YX,YY,YZ,ZX,ZY, ZZ, II$}. The reduced echelon form (neglecting rows of all zeros) for mutual tight inequalities (in $\MP$) of both $v$ and (a) $w_{0}$, (b) $w_{1}$, (c) $w_{2}$, (d) $w_{3}$. For some $x\in \RR^{10}$ each row represents an equality $\sum_{a\in E^{(\text{nl})}} c_{a}x_{a} = 0$, where $c_{a}$ is the coefficient appearing in the matrix, $x_{a}$ is a Pauli expectation value. The $II$ column represents the constant term in each equality.
}
\label{tab:matrix-tight}
\end{table}

Since for each {$w_i$} 
%$\bar{u}_{i}$ 
there are 
%not too many (at most $4$) 
{at most $4$}
CHSH inequalities to consider we can directly calculate $\text{rank}({\bar M}[\zZ_{v}\cup\zZ_{w_i}])$. An increase of rank occurs whenever there are non-trivial equations to be satisfied between the dependent (or non-pivot) columns in Table~\ref{tab:matrix-tight}. Let us see how this works for {$w_0$.}
%$\bar{u}_{0}$. 
The other cases work similarly. Consider the first inequality from {part (1) of Lemma \ref{lem:tight CHSH}}
%Table~\ref{tab:tight-u0}, 
which imposes the constraint
{
\begin{eqnarray}
2 + (w_0)_{YY}   + (w_0)_{YZ}  + (w_0)_{ZY}  - (w_0)_{ZZ} = 0.\notag 
\end{eqnarray}
}
%\begin{eqnarray}
%2 + \langle YY \rangle  + \langle YZ \rangle  + \langle ZY \rangle  - \langle ZZ \rangle = 0.\notag 
%\end{eqnarray}
From the matrix in Table~\ref{tab:rref-u0} we know that 
$(w_0)_{YY}= - (w_0)_{ZZ}$
%$\langle YY \rangle =  -\langle ZZ \rangle$ 
and that 
$(w_0)_{YZ}=(w_0)_{ZY}$.
%$\langle YZ \rangle = \langle ZY \rangle$.
 Up to an overall constant factor this yields the non-redundant equation
\begin{eqnarray}
1 + (w_0)_{ZY} - (w_0)_{ZZ} = 0.
\end{eqnarray}
%\begin{eqnarray}
%1 + \langle ZY \rangle - \langle ZZ \rangle = 0.
%\end{eqnarray}
Since the joint rank, $\text{rank}_{\bar{M}}(v,w_{0}) = 7$, we now have that $\text{rank}(M[\zZ_{v}\cup\zZ_{w_i}]) \geq 8$. It turns out, however, that the other constraints are redundant. For example we have
\begin{eqnarray}
2 - (w_0)_{YX}  + (w_0)_{YY}   + (w_0)_{ZX}   - (w_0)_{ZY}  = 0.\notag
\end{eqnarray}
%\begin{eqnarray}
%2 - \langle YX \rangle  + \langle YY \rangle  + \langle ZX \rangle  - \langle ZY \rangle = 0.\notag
%\end{eqnarray}
Using that 
$(w_0)_{YX}=(w_0)_{YZ}+1$
%$\langle YX \rangle =  \langle ZY \rangle + 1$
and that 
$(w_0)_{ZX}=(w_0)_{ZZ}-1$
%$\langle ZX \rangle =  \langle ZZ \rangle - 1$
together with 
$(w_0)_{YY}=-(w_0)_{ZZ}$
%$\langle YY \rangle = -\langle ZZ\rangle$ 
we obtain the trivial equation $0 = 0$. Following a similar calculation for the other non-neighbors, the non-redundant new constraints for each {$w_i$}
%$\bar{u}_{i}$ 
($i = 0,1,2,3$) can be summarized as:
\begin{eqnarray}
w_{0}:&\quad & \quad 1 + (w_0)_{ZY} - (w_0)_{ZZ} = 0,\\
w_{1}:&\quad & \quad (w_1)_{ZX} = (w_1)_{ZY} = 0,\\
w_{2}:&\quad & \quad (w_2)_{ZX} = 0, ~~ (w_2)_{ZY} = (w_2)_{ZZ} -1,\\
w_{3}:&\quad & \quad (w_3)_{YZ} = (w_3)_{ZX}, ~ (w_3)_{ZY} = 0.
\end{eqnarray}
%\begin{eqnarray}
%w_{0}:&\quad & \quad 1 + \langle ZY \rangle - \langle ZZ \rangle = 0,\\
%w_{1}:&\quad & \quad \langle ZX \rangle = \langle ZY \rangle = 0,\\
%w_{2}:&\quad & \quad \langle ZX \rangle = 0, ~~ \langle ZY \rangle = \langle ZZ \rangle -1,\\
%w_{3}:&\quad & \quad \langle YZ \rangle = \langle ZX \rangle, ~ \langle ZY \rangle = 0.
%\end{eqnarray}
Notice that since the joint rank of {$v$ and} $w_{i}$ ($i=1,2,3$) is $6$, two non-redundant constraints are needed to sufficiently bump the rank up.

It is now appropriate to implement %Proposition~\ref{pro:dd-lemma-adj} 
{Proposition \ref{pro:DD for intersection of pair of polytopes}}
and thus {Eq.~(\ref{eq:ui})}. {We obtain t}he coefficients of the newly generated DD vertices.
% is given in Table~\ref{tab:mpb-vertices}.%
}

 The group ${G}$ %\comm{I think we can simply write $G$} 
 is a subgroup of $\text{Aut}({\overline \CL})$, the combinatorial automorphism group of {$\overline \CL$} Therefore this group also acts on the intersection $\overline{\MP}$ by combinatorial automorphisms.  
{Using the ${G}$ action} we can generate the orbits of these vertices.\\

{

\Rem{
{\rm
The vertex $u_1$ described in this theorem and its orbits under ${G}$ (which constitute the $\tau_3$ type of vertices of $\overline\MP$) are precisely the images of the non-local stabilizer states under the $\pi$ map. The convex hull of the non-local stabilizer states is known as the Clifford polytope \cite{buhrman2006new}. Then the projection under $\pi$ gives the polytope $\Conv(\tau_3)$. We have 
$$
\Conv(\tau_3) \subset {\overline\MP} \subset {\overline \CL}.
$$
As an immediate consequence of this observation we obtain that if a quantum state $\rho$ violates a CHSH inequality then it violates a facet of the Clifford polytope. This result is first proved in \cite{howard2012nonlocality} to show that non-locality can be used for benchmarking universal quantum computation. 
}
}

}

\section{New vertices of $\mathbf{\Lambda_{2}}$}\label{sec:lifting}

We consider lifts of the vertices $u_{i}\in \overline{\MP}$ {described in Theorem \ref{thm:vertices of MP bar}}, {by which we mean points
$\tilde u_{i}\in \RR^{15}$ satisfying $M^\loc \tilde u_{i}\geq I_{36\times 1}$ and mapping to $u_i$ under the map $\pi:\RR^{15}\to \RR^{9}$.}
%a map $\RR^{9}\to \RR^{15}$ such that the lifted point $\tilde u_{i}\in \RR^{15}$ still satisfies $M^\loc \tilde u_{i}\geq I_{36\times 1}$.
 The fact that $\overline{\MP}$ is contained in $\overline\CL$ allows for a simplified method of lifting that ensures that $\tilde u_{i} \in \Lambda_{2}\cap \CL$ and, in particular, such ``classical" lifts yield vertices of $\Lambda_{2}$.

\subsection{{Vertices the classical polytope}
%The $\mathbf{V}$-representation of $\mathbf{\bar C}$
}

%\si{[[Should we add this information earlier in Section~\ref{sec:lambda}?]]} \comm{In the earlier section our focus was the hyperplanes but now we switch to vertices. I think it is ok.}\\

{Recall that vertices of $\CL$ are described by Eq.~(\ref{eq:deterministic-vertices}). It is determined by the values $s_i,r_i\in \ZZ_2$ where $i=0,1,2$. We begin by an alternative description of these vertices.
Let $\Gamma\subset E$ be a subset.
Let $\chi:\Gamma\to \ZZ_2$ be a function such that $\chi(0)=0$ if $0\in \Gamma$. We will say that $\chi$ is even (odd) if $\chi(a+b)=\chi(a)+\chi(b)$ ($\chi(a+b)=\chi(a)+\chi(b)+1$) for all $a,b\in E$ such that $\omega(a,b)=0$. 
Now, a vertex of $\CL$ is given by $d^\xi$ where $\xi:E\to \ZZ_2$ is a function such that $\xi|_{E^\loc}$ is an even function. The coordinates of the vertex are given by 
$$(d^\xi)_a=(-1)^{\xi(a)}.$$ 
Observe that for such a $\xi$ the restriction $\xi|_{E^\nloc}$ is either even (parity $0$) or odd (parity $1$).
Let $\bar\xi$ denote the function obtained by $s_i\mapsto s_i+1$ and $r_i\mapsto r_i+1$. Then $\bar\xi|_{E^\nloc} = \xi|_{E^\nloc}$. Therefore $d^\xi$ and $d^{\bar \xi}$ are mapped to the same point under $\pi$. We will write $\delta^\xi$ for the point $\pi(d^\xi)$.
}

For any $x \in \overline{\MP}$ we {have} the decomposition
\begin{eqnarray}
{x} = \sum_{\xi}q_{\xi}{\delta}^{\xi}\quad \text{where}\quad q_{\xi}\geq 0,~~\sum_{\xi}q_{\xi} = 1.%
\label{eq:convex-bar_C}
\end{eqnarray}
{where $\xi$ runs over functions whose restriction to the local part is even.}
Supposing that a suitable set of parameters $q^{\xi}$ are given, then there is a rather simple (although not the most general) prescription for lifting ${x}$. Supposing we have {chosen} a lift $\tilde{\delta}^{\xi}$ of ${\delta}^{\xi}$, let us define the \textit{classical} lift of ${x}$ to be
\begin{eqnarray}
 \tilde{x} := \sum_{\xi}q_{\xi}\tilde{\delta}^{\xi}.\label{eq:classical-lift-a}
\end{eqnarray}
This is a convex mixture with the same parameters $q_{\xi}$ as in Eq.~(\ref{eq:convex-bar_C}) but with ${\delta}^{\xi}$ replaced with $\tilde{\delta}^{\xi}$.

\begin{lem}
The lift $\tilde \delta^{\xi}$ of $\delta^{\xi}$ is given by
\begin{eqnarray}
\tilde \delta^{\xi} = \alpha_{\xi}\, d^{\xi}+(1-\alpha_{\xi})\, d^{\bar \xi}\quad \text{where}\quad \alpha_{\xi} \in [0,1].
\end{eqnarray}
\end{lem}
\begin{proof}
This follows from the fact that there are two distinct vertices $d^{\xi},d^{\bar \xi}\in D$ such that $\pi(d^{\xi}) = \pi(d^{\bar \xi}) = \delta^{\xi}$.
\end{proof}

\begin{rem}
{\rm
For the pair of vertices $d^{\xi}$ and $d^{\bar \xi}$ we fix the convention that between 
$$z=(r_0,r_1,r_2,s_0,s_1,s_2)\;\text{ and }\;\bar z=(r_0+1,r_1+1,r_2+1,s_0+1,s_1+1,s_2+1)$$  
$z$ is the bit string with least Hamming weight and if %it is 
{the Hamming weights are}
equal then $\bar z$ is assigned to the bit string with the first non-zero entry. {For example, $z=00$ and $\bar z=11$, or when the weights are the same such as $z=01$ and $\bar z =10$.}
}
\end{rem}

The expression in Eq.~(\ref{eq:classical-lift-a})   takes the form
\begin{eqnarray}
 \tilde{x} = \sum_{\xi}q_{\xi}\left (\alpha_{\xi}\, d^{\xi}+(1-\alpha_{\xi})\, d^{\bar \xi}\right ).\label{eq:classical-lift-b}
\end{eqnarray}
Once the $q_{\xi}$'s are fixed by Eq.~(\ref{eq:convex-bar_C}), the only remaining free parameters are the $\alpha_{\xi}$'s which are all in the range $\alpha_{\xi}\in [0,1]$. Thus the space of lifts forms a hypercube in a dimension determined by {the} number of terms appearing in Eq.~(\ref{eq:convex-bar_C}). Furthermore, it is clear that the lift $\tilde{\xi}$ will always, by construction, be in $\Lambda_{2} \cap \CL $, so long as $x\in \overline{\MP}$.

\subsection{Classical lifts}

Finding the coefficients $q_{\xi}$ is itself a {polytope-theoretic} problem. Let us see how this works. Working in homogenized coordinates, first introduce the matrix $R\in \RR^{10\times 32}$ whose columns are the rays ${\delta}^{\xi}$ and let $q \in \RR^{32}$ be a vector. For some $x \in {C(\overline{\MP})}$ we wish to find solutions to the following set of equations
\begin{eqnarray}
x = R\,q\quad\text{subject to}\quad q\geq 0.
\label{eq:convex-decomposition-a}
\end{eqnarray}
The space of feasible vectors $q$ is the positive orthant of $\RR^{32}$ intersected by $10$ constraints, one of which is normalization of the $q_{\xi}$'s, while the others are of the form:
\begin{eqnarray}
(x)_{i} = (R\,q)_{i} =  \sum_{\xi}q_{\xi}({\delta}^{\xi})_{i}\quad (i=1,\cdots,9).
\end{eqnarray}

\begin{table}[h!]
\centering
\begin{subfigure}{.33\textwidth}
  \centering
   {\footnotesize
   \begin{eqnarray}
\begin{tabular}{ c|c c c } 
 $1$ & ~~ & ~~ & ~~ \\
 \hline 
 ~~ & $-1$ & $0$ & $0$ \\ 
 ~~ & $0$ & $0$ & $0$  \\
 ~~ & $0$ & $0$ & $0$
\end{tabular}\notag
\end{eqnarray}
  }%close font
  \caption{}
\end{subfigure}%
\begin{subfigure}{.33\textwidth}
  \centering
   {\footnotesize
   \begin{eqnarray}
\begin{tabular}{ c|c c c } 
 $1$ & ~~ & ~~ & ~~ \\
 \hline 
 ~~ & $0$ & $0$ & $0$ \\ 
 ~~ & $0$ & $1$ & $0$  \\
 ~~ & $0$ & $0$ & $0$
\end{tabular}\notag
\end{eqnarray}
  }%close font
  \caption{}
\end{subfigure}%
\begin{subfigure}{.33\textwidth}
  \centering
   {\footnotesize
   \begin{eqnarray}
\begin{tabular}{ c|c c c } 
 $1$ & ~~ & ~~ & ~~ \\
 \hline 
 ~~ & $0$ & $0$ & $0$ \\ 
 ~~ & $0$ & $0$ & $0$  \\
 ~~ & $0$ & $0$ & $-1$
\end{tabular}\notag
\end{eqnarray}
  }%close font
  \caption{}
\end{subfigure}%

\begin{subfigure}{.99\textwidth}
  \centering
   {\footnotesize
   \begin{eqnarray}
\begin{tabular}{ c|c c c } 
 $1$ & ~~ & ~~ & ~~ \\
 \hline 
 ~~ & $-1$ & $0$ & $0$ \\ 
 ~~ & $0$ & $0$ & $0$  \\
 ~~ & $0$ & $0$ & $0$
\end{tabular}\hspace{2.75 cm}%
\begin{tabular}{ c|c c c } 
 $1$ & ~~ & ~~ & ~~ \\
 \hline 
 ~~ & $0$ & $0$ & $0$ \\ 
 ~~ & $0$ & $1$ & $0$  \\
 ~~ & $0$ & $0$ & $0$
\end{tabular}\hspace{2.75 cm}%
\begin{tabular}{ c|c c c } 
 $1$ & ~~ & ~~ & ~~ \\
 \hline 
 ~~ & $0$ & $0$ & $0$ \\ 
 ~~ & $0$ & $0$ & $0$  \\
 ~~ & $0$ & $0$ & $-1$
\end{tabular}\notag
\end{eqnarray}
  }%close font
  \caption{}
\end{subfigure}%
\caption{Inequalities that enforce the condition that Pauli coefficients satisfy $-1 \leq \langle A \rangle \leq 1 $. We depict the tight inequalities for vertices (a) $u_{0}$,  (b) $u_{2}$,  (c) $u_{3}$,  (d) $u_{1}$. 
}
\label{tab:tight-cube-ui}
\end{table}

This problem can be simplified significantly by recognizing that we are interested in decompositions of vertices ${u}_{i}\in \overline{\MP}$, which by construction are tight at some subset of CHSH facets. More explicitly, if ${u}_{i}$ is tight at some facet ${h}$, then we have that
\begin{eqnarray}
{h}\cdot {u}_{i} = \sum_{z} q_{\xi}({h}\cdot {\delta}^{\xi}) = 0.
\end{eqnarray}
Since $h\cdot \delta^{\xi}\geq 0$ as well as $q_{\xi}\geq 0$, this implies that only rays ${\delta}^{\xi}$ that are also tight at ${h}$ should be involved in the decomposition.

\begin{table}[h!]
\centering
\begin{subfigure}{.39\textwidth}
  \centering
   {\footnotesize
   \begin{eqnarray}
&&\begin{tabular}{ c|c c c } 
 $2$ & ~~ & ~~ & ~~ \\
 \hline 
 ~~ & $-1$ & $-1$ & $0$ \\ 
 ~~ & $-1$ & $1$ & $0$  \\
 ~~ & $0$ & $0$ & $0$
\end{tabular}~~~
\begin{tabular}{ c|c c c } 
$2$ & ~~ & ~~ & ~~ \\
 \hline 
 ~~ & $0$ & $0$ & $0$ \\ 
 ~~ & $0$ & $1$ & $1$  \\
 ~~ & $0$ & $1$ & $-1$
\end{tabular}\notag\\
&&\begin{tabular}{ c|c c c } 
$2$ & ~~ & ~~ & ~~ \\
 \hline 
 ~~ & $0$ & $0$ & $0$ \\ 
 ~~ & $-1$ & $1$ & $0$  \\
 ~~ & $1$ & $1$ & $0$
\end{tabular}~~~%
\begin{tabular}{ c|c c c } 
$2$ & ~~ & ~~ & ~~ \\
 \hline 
 ~~ & $-1$ & $1$ & $0$ \\ 
 ~~ & $0$ & $0$ & $0$  \\
 ~~ & $1$ & $1$ & $0$
\end{tabular}\notag\\
&&\begin{tabular}{ c|c c c } 
$2$ & ~~ & ~~ & ~~ \\
 \hline 
 ~~ & $-1$ & $0$ & $-1$ \\ 
 ~~ & $0$ & $0$ & $0$  \\
 ~~ & $1$ & $0$ & $-1$
\end{tabular}~~~
\begin{tabular}{ c|c c c } 
$2$ & ~~ & ~~ & ~~ \\
 \hline 
 ~~ & $-1$ & $0$ & $-1$ \\ 
 ~~ & $-1$ & $0$ & $1$  \\
 ~~ & $0$ & $0$ & $0$
\end{tabular}\notag
\end{eqnarray}
  }%close font
  \caption{}\label{tab:tight-chsh-u0}
\end{subfigure}%
\begin{subfigure}{.2\textwidth}
  \centering
   {\footnotesize
   \begin{eqnarray}
&&\begin{tabular}{ c|c c c } 
$2$ & ~~ & ~~ & ~~ \\
 \hline 
 ~~ & $-1$ & $\pm 1$ & $0$ \\ 
 ~~ & $\pm1$ & $1$ & $0$  \\
 ~~ & $0$ & $0$ & $0$
\end{tabular}\notag\\
&&\begin{tabular}{ c|c c c } 
$2$ & ~~ & ~~ & ~~ \\
 \hline 
 ~~ & $0$ & $0$ & $0$ \\ 
 ~~ & $0$ & $1$ & $\pm 1$  \\
 ~~ & $0$ & $\pm 1$ & $-1$
\end{tabular}\notag\\
&&\begin{tabular}{ c|c c c } 
$2$ & ~~ & ~~ & ~~ \\
 \hline 
 ~~ & $-1$ & $0$ & $\pm 1$ \\ 
 ~~ & $0$ & $0$ & $0$  \\
 ~~ & $\mp 1 $ & $0$ & $-1$
\end{tabular}\notag
\end{eqnarray}
  }%close font
  \caption{}
\end{subfigure}%
\begin{subfigure}{.2\textwidth}
  \centering
   {\footnotesize
   \begin{eqnarray}
&&\begin{tabular}{ c|c c c } 
$2$ & ~~ & ~~ & ~~ \\
 \hline 
 ~~ & $-1$ & $-1$ & $0$ \\ 
 ~~ & $-1$ & $1$ & $0$  \\
 ~~ & $0$ & $0$ & $0$
\end{tabular}\notag\\
&&\begin{tabular}{ c|c c c } 
$2$ & ~~ & ~~ & ~~ \\
 \hline 
 ~~ & $0$ & $0$ & $0$ \\ 
 ~~ & $0$ & $1$ & $1$  \\
 ~~ & $0$ & $1$ & $-1$
\end{tabular}\notag\\
&&\begin{tabular}{ c|c c c } 
$2$ & ~~ & ~~ & ~~ \\
 \hline 
 ~~ & $0$ & $0$ & $0$ \\ 
 ~~ & $-1$ & $1$ & $0$  \\
 ~~ & $1$ & $1$ & $0$
\end{tabular}\notag
\end{eqnarray}
  }%close font
  \caption{}
\end{subfigure}%
\begin{subfigure}{.2\textwidth}
  \centering
   {\footnotesize
   \begin{eqnarray}
&&\begin{tabular}{ c|c c c } 
$2$ & ~~ & ~~ & ~~ \\
 \hline 
 ~~ & $-1$ & $-1$ & $0$ \\ 
 ~~ & $-1$ & $1$ & $0$  \\
 ~~ & $0$ & $0$ & $0$
\end{tabular}\notag\\
&&\begin{tabular}{ c|c c c } 
$2$ & ~~ & ~~ & ~~ \\
 \hline 
 ~~ & $0$ & $0$ & $0$ \\ 
 ~~ & $0$ & $1$ & $1$  \\
 ~~ & $0$ & $1$ & $-1$
\end{tabular}\notag\\
&&\begin{tabular}{ c|c c c } 
$2$ & ~~ & ~~ & ~~ \\
 \hline 
 ~~ & $-1$ & $0$ & $-1$ \\ 
 ~~ & $0$ & $0$ & $0$  \\
 ~~ & $1$ & $0$ & $-1$
\end{tabular}\notag
\end{eqnarray}
  }%close font
  \caption{}
\end{subfigure}
\caption{CHSH inequalities which are tight for (a) $u_{0}$, (b) $u_{1}$, (c) $u_{2}$, (d) $u_{3}$.
}
\label{tab:tight-chsh-ui}
\end{table}

\begin{lem}\label{lem:tight-ui}
For vertices $u_{i}$ ($i = 0,\cdots,3$) {introduced in Eq.~(\ref{eq:ui})} let $H^{(i)}_{2}\subset H_{2}$ be a subset of facets such that ${h}\in H_{2}^{(i)}$ satisfies $h\cdot u_{i} = 0$. The subsets $H_{2}^{(i)}$ consist of NN inequalities given in Table~\ref{tab:tight-cube-ui} and CHSH inequalities given in Table~\ref{tab:tight-chsh-ui}.
\end{lem}
\begin{proof}See Appendix~\ref{sec:proof-Lem-tight}.
\end{proof}

Let us introduce a subset of vertices ${V}^{(i)}_{2} \subset V_{2}$ such that for all $\delta \in {V}^{(i)}_{2}$ and all $h\in H_{2}^{(i)}$ we have that $h\cdot \delta = 0$. Define now the matrix $R_{2}^{(i)}\in \RR^{10\times \kappa_{i}}$ whose columns are the elements of $V_{2}^{(i)}$ and where $\kappa_{i} := |V_{2}^{(i)}|$.

\begin{table}[h!]
\centering
\begin{subfigure}{.49\textwidth}
  \centering
   {\footnotesize
\begin{eqnarray}
% \begin{matrix} II \\ XX \\ XY \\ XZ \\ YX \\ YY \\ YZ \\ ZX \\ ZY \\ ZZ 
%  \end{matrix}\hspace{1 em}%
&&\begin{bmatrix}
 ~1 & ~1 & ~1 & ~1 \\
 ~1 & ~1 & ~1 & ~1 \\
 ~1 & ~1 & -1 & -1 \\
 ~1 & -1 & ~1 & -1 \\
 -1 & ~1 & ~1 & ~1 \\
 -1 & ~1 & -1 & -1 \\
 -1 & -1 & ~1 & -1 \\
 -1 & -1 & ~1 & -1 \\
 -1 & -1 & -1 & ~1 \\
 -1 & ~1 & ~1 & ~1 
  \end{bmatrix}\notag\\%
\text{parity}&&~~~\begin{matrix} ~0 & ~~0 & ~~1 & ~~1 \end{matrix}\notag%
\notag
\end{eqnarray}
  }%close font
  \caption{}
  \label{tab:R-u0}
\end{subfigure}%
\begin{subfigure}{.49\textwidth}
  \centering
   {\footnotesize
\begin{eqnarray}
% \begin{matrix} II \\ XX \\ XY \\ XZ \\ YX \\ YY \\ YZ \\ ZX \\ ZY \\ ZZ 
%  \end{matrix}\hspace{1 em}%
&&\begin{bmatrix}
1 & 1 & 1 & 1 \\ 
 1 & 1 & 1 & 1 \\ 
 1 & -1 & -1 & 1 \\ 
 1 & 1 & -1 & -1 \\ 
 -1 & 1 & 1 & -1 \\ 
 -1 & -1 & -1 & -1 \\ 
 -1 & 1 & -1 & 1 \\ 
 1 & 1 & -1 & -1 \\ 
 1 & -1 & 1 & -1 \\ 
 1 & 1 & 1 & 1 
  \end{bmatrix}\notag\\%
  \text{parity}&&~~~\begin{matrix} ~1 & ~~1 & ~~1 & ~~1 \end{matrix}%
  \notag
\end{eqnarray}
  }%close font
  \caption{}
\end{subfigure}%

\begin{subfigure}{.49\textwidth}
  \centering
   {\footnotesize
\begin{eqnarray}
% \begin{matrix} II \\ XX \\ XY \\ XZ \\ YX \\ YY \\ YZ \\ ZX \\ ZY \\ ZZ 
%  \end{matrix}\hspace{1 em}%
&&\begin{bmatrix}
1 & 1 & 1 & 1 & 1 & 1 & 1 & 1 \\ 
 1 & -1 & -1 & 1 & -1 & 1 & 1 & 1 \\ 
 1 & 1 & 1 & -1 & 1 & -1 & -1 & 1 \\ 
 1 & -1 & 1 & -1 & 1 & 1 & -1 & -1 \\ 
 -1 & 1 & 1 & 1 & 1 & 1 & 1 & -1 \\ 
 -1 & -1 & -1 & -1 & -1 & -1 & -1 & -1 \\ 
 -1 & 1 & -1 & -1 & -1 & 1 & -1 & 1 \\ 
 -1 & 1 & -1 & 1 & 1 & 1 & -1 & -1 \\ 
 -1 & -1 & 1 & -1 & -1 & -1 & 1 & -1 \\ 
 -1 & 1 & 1 & -1 & -1 & 1 & 1 & 1  \notag
  \end{bmatrix}\notag\\%
\text{parity}&&~~~\begin{matrix} ~0 & ~~0 & ~~0 & ~~~0 & ~~1 & ~~1 & ~~1 & ~~1 \end{matrix}%
\notag
\end{eqnarray}
  }%close font
  \caption{}
\end{subfigure}%
\begin{subfigure}{.49\textwidth}
  \centering
   {\footnotesize
\begin{eqnarray}
% \begin{matrix} II \\ XX \\ XY \\ XZ \\ YX \\ YY \\ YZ \\ ZX \\ ZY \\ ZZ 
%  \end{matrix}\hspace{1 em}%
&&\begin{bmatrix}
1 & 1 & 1 & 1 & 1 & 1 \\ 
 -1 & 1 & 1 & 1 & 1 & 1 \\ 
 1 & 1 & 1 & -1 & -1 & 1 \\ 
 1 & -1 & 1 & 1 & -1 & -1 \\ 
 1 & 1 & -1 & 1 & 1 & -1 \\ 
 -1 & 1 & -1 & -1 & -1 & -1 \\ 
 -1 & -1 & -1 & 1 & -1 & 1 \\ 
 -1 & -1 & 1 & 1 & -1 & -1 \\ 
 1 & -1 & 1 & -1 & 1 & -1 \\ 
 1 & 1 & 1 & 1 & 1 & 1 \notag
  \end{bmatrix}\notag\\%
\text{parity}&&~~~\begin{matrix} ~0 & ~~0 & ~~1 & ~~~1 & ~~1 & ~~1 \end{matrix}%
\notag
\end{eqnarray}
  }%close font
  \caption{}
\end{subfigure}%
\caption{The matrices $R_{2}^{(i)}$ for the vertices (a) $u_{0}$,  (b) $u_{1}$,  (c) $u_{2}$,  (d) $u_{3}$. The rows are labeled by Pauli operators in lexicographic order, i.e. 
%$II$, $XX$, $XY$, $\cdots$, $ZZ$. 
{$II,XX, XY, \cdots, ZZ$.}
At the bottom of each column corresponding to $\delta^{\xi}$ we give the parity of  $\xi$.
}
\label{tab:tight-vertices}
\end{table}%

\begin{pro}\label{pro:decomposition}
The following hold:
\begin{enumerate}
\item The matrices $R_{2}^{(i)}$ are given in Table~\ref{tab:tight-vertices}.
\item A subset of solutions to the system of equations
\begin{eqnarray}
u_{i} = R_{2}^{(i)}q_{i} \quad \text{subject to} \quad q_{i}\geq 0,%
\label{eq:convex-decomposition-b}
\end{eqnarray}
where $q_{i}\in \RR^{\kappa_{i}}$, is given by
\begin{subequations}
\begin{eqnarray}
q_{0}^{T} &=& \begin{bmatrix} 1/4 & 1/4 & 1/4 & 1/4\end{bmatrix}\label{eq:decomp-u0}\\%
q_{1}^{T} &=& \begin{bmatrix} 1/4 & 1/4 & 1/4 & 1/4\end{bmatrix}\label{eq:decomp-u1}\\%
q_{2}^{T} &=& \begin{bmatrix} 0 & 0 & 0 & 0 & 1/4 & 1/4 & 1/4 & 1/4 \end{bmatrix}\label{eq:decomp-u2}\\%
q_{3}^{T} &=& \begin{bmatrix} 1/6 & 1/6 & 1/6 & 1/6 & 1/6 & 1/6\end{bmatrix}\label{eq:decomp-u3}.
\end{eqnarray}
\end{subequations}
\end{enumerate}
\end{pro}

\begin{proof}
The proof can be found in Appendix~\ref{sec:proof-Pro-decomp}.
\end{proof}

%\si{[[Sketch of proof to be added in Appendix.]]}

\begin{rem}
{\rm%
Equations (\ref{eq:decomp-u0}),  (\ref{eq:decomp-u1}), and (\ref{eq:decomp-u3}) are the 
%\textit{unique}
unique
 solution to the system of equations in (\ref{eq:convex-decomposition-b}), while Eq.~(\ref{eq:decomp-u2}) is not; although all three solutions are uniform mixtures of four vertices.
}
\end{rem}

\begin{table}[h!]
\centering
\begin{subfigure}{.49\textwidth}
  \centering
   {\footnotesize
\begin{eqnarray}
% \begin{matrix} II \\ XX \\ XY \\ XZ \\ YX \\ YY \\ YZ \\ ZX \\ ZY \\ ZZ 
%  \end{matrix}\hspace{1 em}%
&&\begin{bmatrix}
 ~1 & ~1 & ~1 & ~1 \\
 -1 & -1 & ~1 & -1 \\
 ~1 & -1 & ~1 & -1 \\
 ~1 & ~1 & ~1 & ~1 \\
 -1 & -1 & ~1 & -1 \\
 -1 & -1 & -1 & ~1 \\
 -1 & ~1 & ~1 & ~1  
  \end{bmatrix}\notag\\%
\text{parity}&&~~~\begin{matrix} ~0 & ~~0 & ~~1 & ~~1 \end{matrix}\notag%
\notag
\end{eqnarray}
  }%close font
  \caption{}
  \label{tab:R-u0}
\end{subfigure}%
\begin{subfigure}{.49\textwidth}
  \centering
   {\footnotesize
\begin{eqnarray}
% \begin{matrix} II \\ XX \\ XY \\ XZ \\ YX \\ YY \\ YZ \\ ZX \\ ZY \\ ZZ 
%  \end{matrix}\hspace{1 em}%
&&\begin{bmatrix}
 ~1 & ~1 & ~1 & ~1 \\ 
 -1 & -1 & ~1 & -1 \\ 
 ~1 & -1 & ~1 & ~1 \\ 
 -1 & -1 & -1 & ~1 \\ 
 -1 & -1 & ~1 & -1 \\ 
 -1 & ~1 & -1 & -1 \\ 
 -1 & -1 & -1 & ~1
  \end{bmatrix}\notag\\%
  \text{parity}&&~~~\begin{matrix} ~1 & ~~1 & ~~1 & ~~1 \end{matrix}%
  \notag
\end{eqnarray}
  }%close font
  \caption{}
\end{subfigure}%

\begin{subfigure}{.49\textwidth}
  \centering
   {\footnotesize
\begin{eqnarray}
% \begin{matrix} II \\ XX \\ XY \\ XZ \\ YX \\ YY \\ YZ \\ ZX \\ ZY \\ ZZ 
%  \end{matrix}\hspace{1 em}%
&&\begin{bmatrix}
 ~1 & ~1 & ~1 & ~1 \\ 
 ~1 & -1 & -1 & -1 \\ 
 -1 & -1 & -1 & ~1 \\ 
 -1 & -1 & ~1 & ~1 \\ 
 -1 & -1 & -1 & -1 \\ 
 ~1 & ~1 & ~1 & -1 \\ 
 ~1 & -1 & ~1 & ~1
 \end{bmatrix}\notag\\%
\text{parity}&&~~~\begin{matrix} ~~1 & ~~1 & ~~1 & ~~1 \end{matrix}%
\notag
\end{eqnarray}
  }%close font
  \caption{}
\end{subfigure}%
\begin{subfigure}{.49\textwidth}
  \centering
   {\footnotesize
\begin{eqnarray}
% \begin{matrix} II \\ XX \\ XY \\ XZ \\ YX \\ YY \\ YZ \\ ZX \\ ZY \\ ZZ 
%  \end{matrix}\hspace{1 em}%
&&\begin{bmatrix}
 ~1 & ~1 & ~1 & ~1 & ~1 & ~1 \\ 
 ~1 & -1 & ~1 & -1 & -1 & -1 \\ 
 -1 & -1 & -1 & -1 & -1 & ~1 \\ 
 ~1 & ~1 & ~1 & -1 & ~1 & ~1 \\ 
 -1 & -1 & ~1 & -1 & -1 & -1 \\ 
 ~1 & -1 & ~1 & ~1 & ~1 & -1 \\ 
 ~1 & ~1 & ~1 & -1 & ~1 & ~1
  \end{bmatrix}\notag\\%
\text{parity}&&~~~\begin{matrix} ~0 & ~~0 & ~~1 & ~~~1 & ~~1 & ~~1 \end{matrix}%
\notag
\end{eqnarray}
  }%close font
  \caption{}
\end{subfigure}%
\caption{
%\si{The matrices $R_{2}^{(i)}$ for the vertices (a) $u_{0}$,  (b) $u_{1}$,  (c) $u_{2}$,  (d) $u_{3}$. The rows are labeled by Pauli operators in lexicographic order, i.e. $II$, $XX$, $XY$, $\cdots$, $ZZ$. At the bottom of each column corresponding to $\delta_{z}$ we give the parity of the bit-string $z$.} 
{The local part of the lifts of the functions $\xi$ given by each column  in Table \ref{tab:tight-vertices}. Rows are labeled by 
$II, XI, YI, ZI, IX, IY, IZ$.}
}
\label{tab:lifts}
\end{table}%

Using the decompositions in Eqns.~(\ref{eq:decomp-u0})-(\ref{eq:decomp-u3}) and Eq.~(\ref{eq:convex-decomposition-b}) we can obtain vertices of $\Lambda_{2}$. To do so, it suffices to consider the vertices of the hypercube of classical lifts so that $\alpha_{\xi}\in \{0,1\}$. Compared to a general lifting procedure this is a greatly simplified task as there are only $2^{\kappa_{i}}$ options to check for each $u_{i}$. Representatives of the vertices obtained in this fashion are precisely those given in Eq.~(\ref{eq:vertices-lambda2}) together with the choice of $\alpha = (\alpha_{z_{1}},\cdots,\alpha_{z_{\kappa_{i}}})\in \{0,1\}^{\kappa_{i}}$ for all $\tilde \delta^{\xi}$ appearing in Eq.~(\ref{eq:convex-decomposition-b}) where $q_{\xi}>0$. {We obtain $\tilde u_i$ by lifting
\begin{enumerate}[(a)]
\item $u_{0}$ with $\alpha = \left (0,0,1,1 \right )$,
\item  $u_{1}$ with $\alpha = \left (0,0,1,0 \right )$,
\item $u_{2}$ with $\alpha = \left (1,0,0,0 \right )$,
\item $u_{3}$ with $\alpha = \left (1,0,1,0,0,0 \right )$.
\end{enumerate}}
%Vertices of $\Lambda_{2}$ obtained by a classical lift of (a) $u_{0}$ with $\alpha = \left (0,0,1,1 \right )$ (b) $u_{1}$ with $\alpha = \left (0,0,1,0 \right )$ (c) $u_{2}$ with $\alpha = \left (1,0,0,0 \right )$ (d) $u_{3}$ with $\alpha = \left (1,0,1,0,0,0 \right )$.
{Using the action of the Clifford group $\Cl_2$ we can generate other vertices in the orbit of these four.}

\section{{Update rules}}

In this section we will describe the update rules of the $\Lambda$ vertices given in Eq.~(\ref{eq:vertices-lambda2}) using the convex decompositions {(Proposition \ref{pro:decomposition})} provided in terms of the vertices of $\CL$.  

{More explicitly, let $A_i$ denote the Hermitian operators corresponding to $\tilde u_i$ given in Eq.~(\ref{eq:vertices-lambda2}). Our goal is to compute $\Pi_a^r A_i \Pi_a^r$ where $a\in E-\set{0}$ and $r\in \ZZ_2$.
When $p(a)=\Tr(\Pi_a^r A_i)>0$ we know that the updated vertex $\Pi_a^r A_i \Pi_a^r/p(a)$ falls inside the convex hull of $T_2$ $\Lambda$-vertices. This follows from the $\Phi$-map described in \cite{okay2021extremal}. Furthermore, we know that these $T_2$ cnc vertices are of the form $A_{\Span{a}^\perp}^\gamma$ where  $\Span{a}^\perp=\set{b\in E:\,\omega(b,a)=0}$.
}

For a subset $\Gamma\subset E$ and a function $\chi:\Gamma \to \ZZ_2$
%. If $0\in \Gamma$ 
%we will assume $\chi(0)=0$. 
we will write 
$$
A^\chi_\Gamma = \frac{1}{4}\sum_{a\in \Gamma} (-1)^{\chi(a)} T_a 
$$
for the corresponding Hermitian operator. 
For simplicity of notation we will sometimes write $A^\chi$ omitting   $\Gamma$.
%We say $\chi$ is even (or odd) if for $a,b\in \Gamma$ with $\omega(a,b)=0$ we have $\chi(a+b)=\chi(a)+\chi(b)$ (or $\chi(a+b)=\chi(a)+\chi(b)+1$). A function $\xi:E\to \ZZ_2$ is said to have even (odd) parity if the restriction $\xi|_{E^\nloc}$ is even (odd). 
Recall that a vertex of $\CL$ is given by $d^\xi$ where $\xi:E\to \ZZ_2$ is a function such that $\xi|_{E^\loc}$ is even.
To distinguish these deterministic vertices easily we will write $D^\xi$ for the corresponding operator $A^\xi_E$.
%by the operator $A_E^\xi$ where $\xi$ satisfies the properties that (1) $\xi|_{E^\loc}$ is even and (2) $\xi|_{E^\nloc}$ is either even or odd. The latter property specifies the parity of $\xi$. 
For computing the update rules  we will rely on the following formula:
\begin{equation}\label{eq:update simplified}
\Pi_a^r D^\xi \Pi_a^r = \Pi_a^r A^{\chi} \Pi_a^r
\end{equation}
where $\chi$ is the restriction $\xi|_{\Span{a}^\perp}$.
%That is, given $\xi$ we need to analyze the operator corresponding to the restriction of this function on the subspace $\Span{a}^\perp$.  
%\comm{introduce $\Span{a}^\perp$} 

For $a\in E^\loc$, we have 
\begin{equation}\label{eq:a perp}
\Span{a}^\perp = \Span{a,c} \cup \Span{a,c_1}\cup \Span{a,c_2}
\end{equation}
where $c,c_2,c_2\in E^\loc$. Given a deterministic vertex {$D^\xi$} the restriction $\xi|_{\Span{a}^\perp}$ is always even. Then using Eq.~(\ref{eq:update simplified}) we obtain the following update rule.

\Lem{\label{lem:local update}
Let $a\in E^\loc$ and $D^\xi$ be a deterministic vertex. Then
$$
\Pi_a^r {D^\xi} \Pi_a^r = \left\lbrace
\begin{array}{ll}
A^{\chi} & \xi(a)=r \\
\zero & \xi(a)=r+1
\end{array}
\right.
$$
{where $\chi$ is the restriction $\xi|_{\Span{a}^\perp}$.}
}

The update of $A_i$ under a local Pauli measurement can be computed from the uniform decomposition into the deterministic vertices. For example, for $T_{a}=X\otimes I$ we have
$$
\Pi_a^0 A_i \Pi_a^0 = \left\lbrace
\begin{array}{ll}
\frac{1}{4} A^{\chi_0^{(3)}} & i=0\\
\frac{1}{4} A^{\chi_1^{(3)}} & i=1\\
\frac{1}{4} A^{\chi_2^{(1)}} & i=2\\
\frac{1}{4} (A^{\chi_3^{(1)}}+ A^{\chi_3^{(3)}}) & i=3
\end{array}
\right. 
$$
where $\chi_i^{(j)}:\Span{a}^\perp\to \ZZ_2$ is the even function obtained by restricting the function $\xi_i^{(j)}$ with $\xi_i^{(j)}(a)=0$. Here the index $j$ corresponds to the column number in Table \ref{tab:lifts}.

For $a\in E^\nloc$, the subspaces in Eq.~(\ref{eq:a perp}) satisfy $c\in E^\loc$ and $c_1,c_2\in E^\nloc$. Given a deterministic vertex $D^\xi$ the restriction $\xi|_{\Span{a,c_i}}$ can be either even or odd depending on the parity of $\xi$. The update rules under a non-local Pauli measurement can be described using the following result.

\Lem{\label{lem:local update}
Let $a\in E^\nloc$ and $D^\xi$ be a deterministic vertex. Assume that $\xi$ has parity $t$ and $\beta(a,c_1)=t$. Then
$$
\Pi_a^r D^\xi \Pi_a^r = \left\lbrace
\begin{array}{ll}
\frac{1}{2}(A^{\gamma_0}+A^{\gamma_1} ) & \xi(a)=r \\
\frac{1}{4}((-1)^{\xi(c_2)} T_{c_2} + (-1)^{\xi(a+c_2)}T_{a+c_2}  )  & \xi(a)=r+1
\end{array}
\right.
$$
where $\gamma_i$ is the outcome assignment obtained by extending $\xi|_{\Span{a,c}\cup \Span{a,c_1}}$ to $\Span{a}^\perp$ by setting $\gamma_i(c_2)=i$.
}

For example, for $T_a=X\otimes X$ and $i=0,1$ we have
$$
\begin{aligned}
\Pi_a^0 A_i \Pi_a^0 = \frac{1}{8}\sum_{j=1}^4 (A^{\gamma_0^{(j)}} +  A^{\gamma_1^{(j)}} )
\end{aligned}
$$
since $\xi_i^{(j)}(a)=0$ for all column indices $j$. However, this does not hold for $i=2,3$ and the computation is a bit more subtle in this case.
We have
$$
\begin{aligned}
\Pi_a^0 A_2 \Pi_a^0 &= \frac{1}{4}\sum_{j=1}^4  \Pi_a^0 A^{\xi_2^{(j)}} \Pi_a^0 \\
&=  \frac{1}{4}A^\gamma + \frac{1}{8}(A^{\gamma_0^{(3)}} + A^{\gamma_0^{(3)}})+\frac{1}{8}(A^{\gamma_0^{(4)}} + A^{\gamma_0^{(4)}}).  
\end{aligned}
$$
Note that $A^\gamma$ is obtained by the sum of $\Pi_a^0 A^{\xi_2^{(1)}} \Pi_a^0 +\Pi_a^0 A^{\xi_2^{(2)}} \Pi_a^0$ where the first term corresponds to the $\xi_2(a)=1$ case hence does not split as a uniform mixture of two cnc vertices. The computation is similar for $A_3$. Again the first two updated deterministic vertices combine to give a cnc vertex:
 $$ 
\Pi_a^0 A_3 \Pi_a^0  = \frac{1}{4}A^\gamma + \frac{1}{8}\sum_{j=3}^6( A^{\gamma_0^{(j)}} + A^{\gamma_0^{(j)}} ). 
$$

\section{Conclusion}\label{sec:conclusion}

Currently it is not known whether there exists a vertex enumeration algorithm whose time complexity is polynomial in both the input ($H$-description) and output ($V$-description). While it may be that the structure of $\Lambda$-polytopes lends itself eventually to an efficient characterization there are a few factors which cast some doubt on this possibility. One reason for this is that the dimension of $\Lambda_{n}$, as well as the number of stabilizer states \cite{aaronson2004improved}, scale exponentially in $n$. For example, $\Lambda_{3}$ is a polytope full-dimensional in $\RR^{63}$ and requires $1080$ input stabilizer states to furnish its $H$-representation. The number of input facets together with the high dimensionality render incremental algorithms, such as the DD method, less effective as one encounters an explosion of intermediate vertices to be computed. Indeed, performing a variant of the method explored in this paper we find that within a few iterations of the DD algorithm quickly approaches several tens of thousands of intermediate vertices. At the same time, however, $\Lambda$-polytopes are also highly degenerate, which causes complications for pivot-based enumeration algorithms.

Among the simplifying factors for $\Lambda$ polytopes, and one that we took advantage of here, is that they possess a high degree of symmetry. Future algorithmic explorations of $\Lambda$-polytopes should more systematically exploit this feature; see e.g., \cite{bremner2009polyhedral}. Another piece of structural information that may prove useful is that for all $n$ we have that $\Lambda_{n}$ is a subpolytope of $ \text{NS}_{n32}$, the nonsignaling polytope for the $(n,3,2)$ Bell scenario. For instance, the vertices of $\text{NS}_{322}$ have previously been characterized (see e.g., \cite{pironio2011extremal}) and its vertices can be lifted to those of $\text{NS}_{332}$. This would aid in the construction of an initial DD pair, likely leading to better performance.

   Among the challenging aspects of mathematically characterizing vertices of $\Lambda$ is the eventual need to compute the rank of matrices whose size grows exponentially in $n$. Graph theory provides one avenue for systematically studying this problem. Such methods were, in fact, used in \cite{zurel2023simulation} to capture $T_{4}$ in $\Lambda_{2}$, but are valid for all $n$. However, while their approach proves the existence of such vertices, their proof is not constructive and no prescription is given for computing the rank. {Graph-theoretic} techniques also underlie the computation of joint rank used here (see Appendix~\ref{sec:comb-mp1}), {which fit in the general theory of simplicial distributions \cite{okay2022simplicial},} and one direction of future research is to refine these methods.

\bibliography{bib.bib}
\bibliographystyle{ieeetr}

\appendix

\section{Double description method}
%\section{Double description method for solving the VEP}
\label{sec:DDM-VEP}

\subsection{Polytope theory}\label{sec:polytope-theory}
First let us recall some basic facts from polytope theory. Let $P\subset \mathbb{R}^{d}$ be a polytope defined in its $H$-representation as the intersection of $m$ half-space inequalities, i.e., $P{(M,b)} = \{{x}\in \mathbb{R}^{d}:\,{M}{x}\geq b\}$. By the Minkowski-Weyl theorem \cite{ziegler2012lectures}, this is equivalent to the so-called $V$-representation of the polytope, given by $P = {\conv(V)}$, where $V$ is {a} matrix whose columns are given by a finite {set of points} $v_{i}\in \RR^{d}$.

In the $H$-representation we index the $m$ inequalities by ${[m]}= \{1,\cdots,m\}$. We adopt a notation where if $\zZ\subset {[m]}$ then ${M}[\zZ]$ is the matrix obtained by keeping only those rows indexed by $\zZ$ and discarding the rest, and similarly for $b[\zZ]$. When adding (subtracting) an element $i$ to (from) a set $U$ we use the notation {$U\cup \set{i}$ ($U- \set{i}$).}
% $S\pm i$. 
We call an inequality $i\in {[m]}$ \textit{tight} at ${x}\in P$ if the inequality is satisfied with equality:
$${M}^{i}{x} = b^{i},$$
where we introduce the shorthand ${M}^{i} = {M}[\{i\}]$ and $b^{i} = b[\{i\}]$. A point ${v}\in P$ is then a vertex if and only if it is the unique solution to $d$ tight inequalities. Let $\zZ_{x},\zZ_{x'}\subset {[m]}$ index the tight inequalities of two points ${x},{x}^{\prime}\in P({M},b)$.%
{
\begin{defn}\label{def:joint-rank}
{\rm
The \textit{joint rank} of two points ${x},{x}^{\prime}\in P({M},b)$, denoted $\text{rank}_{{M}}(x,x')$, is given by the rank of the matrix ${M}[\zZ_{x}\cap \zZ_{x^{\prime}}]$.}
\end{defn}
\noindent Two vertices ${v}, {v'} \in P({M},b)$ are called \textit{neighbors} if $\text{rank}_{{M}}(v,v') = d-1$. More generally, we call two vertices $k$-neighbors if $\text{rank}_{{M}}(v,v') = d-k$.
}

\subsection{Mapping to cones}\label{sec:cones}
It is convenient (and standard) to formulate the {double description (DD)} method using cones rather than polytopes directly. Recall that if $P = P({\bar{M}},\bar{b})$ is a polyhedron in $\RR^{d}$, with ${\bar{M}}\in\RR^{m\times d}$ and $\bar{b}\in\RR^{m}$, then we can construct a cone (see e.g., \cite{ziegler2012lectures}) by introducing an additional homogenizing coordinate $x_{0}$ so that
\begin{eqnarray}
C(P) := P({M},\zero_{m})\quad \text{where} \quad%
{M} = \begin{bmatrix} 1 & \zero_{d}^{T}\\ -\bar{b} & {\bar{M}}\end{bmatrix}\notag
\end{eqnarray}
and $\zero_{k}\in \RR^{k}$ is a vector of all zeros. 
%This is the $H$-description of the cone $P^{c}$.
We see that the original polyhedron $P$ is recovered by restricting $x_{0} = 1$,
\begin{eqnarray}
P = \left \{x\in \RR^{d}~:~\begin{bmatrix}1 \\ x \end{bmatrix}\in C(P)\right \}.\notag
\end{eqnarray}
By the Minkowski-Weyl theorem for cones \cite{ziegler2012lectures}, for each cone ${C(P)} = P({M},\zero)$ there is an equivalent $V$-representation. Letting $R \in \RR^{d\times n}$ be a finite point set, consisting of $n$ column vectors, then such a representation is given by positive combinations of the vectors in $R$,
\begin{eqnarray}
{{\cone}(R) :=}   \left \{x = R\lambda ~:~\lambda\in \RR^{n},~\lambda_{i}\geq 0\right \}.\notag
\end{eqnarray}
%Such finitely generated cones will sometimes be denoted $C = \text{cone}(R)$. In particular, 
%{Vertex enumeration for $P$ is equivalent to extreme ray enumeration for $P^c$.}
%{S}uppose that we have a polytope $P$ for which the vertex enumeration problem is solved, so that $P = \text{conv}(V)$ where the columns of $V\in\RR^{d\times n}$ are the vertices of $P$. We have the following basic fact in convex geometry (see e.g., \cite{ziegler2012lectures}) that
%\begin{eqnarray}
%x \in \text{conv}\left ( V \right )\quad \iff\quad \begin{bmatrix} 1 \\ x \end{bmatrix} \in  \text{cone}\begin{bmatrix} {\onee_n^{T}} \\ V \end{bmatrix},%
%\label{eq:conv-cone}
%\end{eqnarray} 
%where {$\onee_n$} is a \textit{row} vector of all ones.
% {and $P^c$ is the cone}. 
%If we define the cone $P^{c} = \text{cone}\left (( \one, V )^{T}\right )$, 
%{T}hen we recover $P = \left \{x\in\RR^{d}~:~ {(1,x)}\in P^{c}\right \}$ \comm{I think the coordinates are understood without the transpose. Otherwise we should write $(1,x^T)^T$.}. 
%Thus we can easily convert between the description of $P$ as a polytope and the description of $P$ as a slice of a cone. {Accordingly, we refer to the vertices of the polytope $P$ and extreme rays of the cone $P^{c}$ interchangeably.}
The pair of matrices $({M},R)$ which ensure the equality of the $H$ and $V$ representations is called a DD pair. It is standard to refer to ${M}$ as the \textit{generating} matrix, while $R$ is called the \textit{representation} matrix. Constructing the representation matrix $R$ from the generating matrix {$M$} is the {extreme ray enumeration problem}, which is the analogue for cones of the vertex enumeration problem for polytopes.

\subsection{The DD algorithm}\label{sec:dd-algorithm}
The DD method proceeds as follows. Suppose we wish to find a DD pair $({M},R)$ for given generating matrix ${M}\in\RR^{m\times d}$ and unknown representation matrix $R\in \RR^{d\times n}$. Letting the rows of ${M}$ be indexed by ${[m]}$, consider the submatrix ${M}[I_{k}]$ of ${M}$ whose rows are indexed by the subset $I_{k}\subset {[m]}$ which consists of $k\leq m$ elements. For notational convenience we let ${M}_{k} := {M}[I_{k}]$. Suppose now that we have a DD pair $({M}_{k},R^{(k)})$, where $R^{(k)}$ is the solution to the extreme ray enumeration problem for the submatrix ${M}_{k}$. Clearly, if $k$ equals $m$ then we are done. If $k$ is less than $m$, then we wish to construct the DD pair $({M}_{k+1},R^{(k+1)})$ from $({M}_{k},R^{(k)})$, where ${M}_{k+1}$ has rows indexed by $I_{k+1} = I_{k}{\cup \set{i}}$ with $i \in {[m]} - I_{k}$. {We will denote by $P_{k}$ the polytope whose $H$-representation is generated by ${M}_{k}$, or equivalently, whose $V$-representation is generated by the positive hull of $R^{(k)}$.}

%There is a geometrical picture which one can attribute to the DD method. Consider, for instance, that the DD pair $(A_{k},R^{(k)})$ defines a cube. The addition of an additional inequality amounts to intersecting the cube with a hyperplane. In general this will result in new extreme points. The DD method is a scheme for identifying these extreme points.
%\begin{figure}[h!]
%\centering
%\includegraphics[width = .4\linewidth]{dd-method}
%\caption{\comm{I think we can remove this picture. The figure in the main text should suffice.}
%}
%\label{fig:dd-method}
%\end{figure}

We wish find the representation matrix $R^{(k+1)}$ from the known generating matrix ${M}_{k+1}$ and the DD pair $({M}_{k},R^{(k)})$, where we let $J_{k}$ index the $|J_{k}|\in\NN$ columns of $R^{(k)}$. The newly introduced inequality ${M}^{i}\, x\geq 0$ partitions the ambient space $\RR^{d}$ into three subspaces:
\begin{eqnarray}
H_{i}^{+} &=& \left \{x\in \RR^{d}~:~{M}^{i}x > 0\right \},\notag\\
H_{i}^{0} &=& \left \{x\in \RR^{d}~:~{M}^{i}x = 0\right \},\\
H_{i}^{-} &=& \left \{x\in \RR^{d}~:~{M}^{i}x < 0\right \}\notag.
\end{eqnarray}
Let us index the columns $r_{j}$ ($j\in J_{k}$) of $R^{(k)}$ according to which subspace they reside in:
\begin{eqnarray}
J_{k}^{+} &=& \left \{j\in J_{k}~:~r_{j} \in H_{i}^{+}\right \},\notag\\
J_{k}^{0} &=& \left \{j\in J_{k}~:~r_{j} \in H_{i}^{0}\right \},\\
J_{k}^{-} &=& \left \{j\in J_{k}~:~r_{j} \in H_{i}^{-}\right \}\notag.
\end{eqnarray}
New extreme rays of $R^{(k+1)}$ are those which lie in $H_{i}^{0}$, which in turn are generated by taking positive (or conic) combinations of extreme rays indexed by $J_{k}^{+}$ and $J_{k}^{-}$. The negative rays lying in $H_{i}^{-}$ are discarded as they do not satisfy the desired inequality. That this is sufficient for constructing $R^{(k+1)}$ is guaranteed by the following proposition.
\begin{pro}{(\cite{fukuda2005double})}\label{pro:dd-lemma}
Let $({M}_{k},R^{(k)})$ be a DD pair and let $i\in {[m]} - I_{k}$, then $({M}_{k+1},R^{(k+1)})$ is a DD pair where $R^{(k+1)} \in \RR^{d\times |J_{k+1}|}$ with column vectors $r_{j}$ ($j\in J_{k+1}$) defined by
\begin{eqnarray}
J_{k+1} &=& J_{k}^{+}\cup J_{k}^{0}\cup (J^{+}_{k}\times J_{k}^{-})\quad \text{and}\notag\\
r_{jj'} &=& \left ({M}_{i}r_{j}\right )r_{j^{\prime}} - \left ({M}_{i}r_{j^{\prime}}\right )r_{j^{}}\quad\text{for all}\quad (j,j^{\prime})\in J_{k}^{+}\times J_{k}^{-}.
\end{eqnarray}
\end{pro}

A downside of the DD method as implied by Proposition~\ref{pro:dd-lemma} is that  many of the new rays generated turn out to be redundant, however, a refinement is possible. Recalling that $\zZ_{r}\subset {[m]}$ indexes the tight inequalities of a ray $r$, we call a ray $r\in {C(P)}$ extreme if $\text{rank}({M}[\zZ_{r}]) = d-1$. Moreover, two rays $r,r^{\prime}\in {C(P)}$ are adjacent extreme rays if $\text{rank}({M}[\zZ_{r}\cap\zZ_{r^{\prime}}]) = d-2$. Modifying Definition~\ref{def:joint-rank} for cones, two rays are $k$-neighbors if their joint rank is $d-k-1$.

We now have the strengthened proposition
\begin{pro}{(\cite{fukuda2005double})}\label{pro:dd-lemma-adj}
Let $({M}_{k},R^{(k)})$ be a DD pair such that $\text{rank}({M}_{k}) = d$ and let $i\in {[m]} - I_{k}$, then $({M}_{k+1},R^{(k+1)})$ is a DD pair where $R^{(k+1)} \in \RR^{d\times |J_{k+1}|}$ with column vectors $r_{j}$ ($j\in J_{k+1}$) defined by
\begin{subequations}
\begin{eqnarray}
J_{k+1} &=& J_{k}^{+}\cup J_{k}^{0}\cup \Adj\quad \text{and}\\
\Adj &=& \left \{ (j,j^{\prime})\in J_{k}^{+}\times J_{k}^{-}~:~\text{rank}({M}_{k}[\zZ_{r_{j}}\cap\zZ_{r_{j^{\prime}}}]) = d-2\right \}\\
r_{jj'} &=& \left ({M}_{i}r_{j}\right )r_{j^{\prime}} - \left ({M}_{i}r_{j^{\prime}}\right )r_{j^{}}\quad\text{for all}\quad (j,j^{\prime})\in \Adj.\label{eq:dd-vertex}
\end{eqnarray}
\end{subequations}
\end{pro}
With Proposition~\ref{pro:dd-lemma-adj} in hand, to construct $R^{(k+1)}$ it suffices to consider conic combinations only of pairs of adjacent extreme rays lying on either side of the newly inserted inequality ${M}^{i}x\geq 0$, rather than all possible positive combinations of extreme rays. This comes at the cost of checking whether extreme rays in question are indeed adjacent or not.

\subsection{Intersection of polytopes}\label{sec:intersection}

Let us first define the intersection of polytopes. Consider a pair of full-dimensional polytopes $P_{i} \subset \RR^{d}$ ($i=1,2$) given in their $H$-representations by the inequalities ${M}_{i}x\geq b_{i}$. Their intersection ${P_{12}} = P_{1}\cap P_{2}$ can be defined as
\begin{eqnarray}
{P_{12}} = \{x\in\mathbb{R}^{d}~:~{M}x\geq b\}\quad\text{where}\quad%
{M} = \begin{bmatrix}{M}_{1}\\ {M}_{2}\end{bmatrix},~%
b = \begin{bmatrix}b_{1}\\ b_{2}\end{bmatrix}.%
\label{eq:intersection}
\end{eqnarray}
If a point $v\in\mathbb{R}^{d}$ is a vertex of $P_{i}$ and if $v$ satisfies the enlarged set of inequalities ${M}x\geq b$, then it will also be a vertex of {$P_{12}$}.

 To get started with the DD method one requires an initial DD pair $({M}_{k},R^{(k)})$. 
In practice one may choose to start with a single inequality, or with any submatrix of ${M}_{k}$ such that $\text{rank}({M}_{k}) = d$. For our purposes, our target polytope is  
{$P_{12}$}, which is given as the intersection of two other polytopes, {$P_1$ and $P_2$}. 
Thus, using {Eq.~(\ref{eq:intersection})}, it is possible for us to take our initial DD pair to be the generating and representation matrices constructed from the $H$ and $V$ representations of {$P_1$}. New vertices will be generated as we sequentially insert  
{the} inequalities defining the facets of {$P_2$}.

{ 
\Proof{[{\bf Proof of Proposition \ref{pro:DD for intersection of pair of polytopes}}]
{Suppose that at a finite stage of the DD algorithm we introduce the inequality $h_v$ and $u=pv+(1-p)w$ is such that $h_v\cdot u=0$, where $p\in [0,1]$ and $w\in V_1$ is a neighbor of $v$.}
%The point $u$ obtained after the insertion of $h_v$ satisfies $h_v\cdot u=0$. 
From {the} equation {$h_v\cdot u=0$} we find that 
$$
p = \frac{1}{1-\frac{h_v\cdot v}{h_v\cdot w}} \in [0,1].
$$
If $w\in \phi(H_2)$ then 
we have
$$ 
h_w\cdot u = \frac{1}{1-\frac{h_v\cdot v}{h_v\cdot w}} h_{w}\cdot v + \frac{1}{1-\frac{h_v\cdot w}{h_v\cdot v}}h_{w}\cdot w \\
=
\frac{(h_v\cdot w)(h_w\cdot v)-(h_v\cdot v)(h_w\cdot w)}{h_v\cdot w - h_v\cdot v} \geq 0.
$$ 
{Here, conditions (1) and (3) imply that the denominator and the numerator are non-negative, respectively.} 
For $v'\in \phi(H_2)-\set{v,w}$  we have
$$
h_{v'}\cdot u = p h_{v'}\cdot v + (1-p)h_{v'}\cdot w \geq 0
$$
by conditions (1) {and (2)}.
Therefore the vertex $u$ survives to the final stage of the DD algorithm giving a vertex of $P_{12}$. 
}
}

\section{Combinatorial structure of the Mermin polytope}\label{sec:comb-mp1}

\subsection{Neighbors}

Two neighbor vertices specify an edge of the polytope.
For the description of neighbors let $\Omega$ be a cnc set. 
Its complement $\Omega^c{=E-\Omega}$ can be regarded as a loop on the Mermin torus; see Fig.~(\ref{fig:mermin-neighbors}). 
%We will write $l_\Omega$ for this loop, i.e., for the set of edges that the loop intersects. 
Then an edge of the graph of the polytope $\MP$ {can be described} by a function 
$$\varphi:\Omega^c \to \ZZ_2$$ 
(We prefer the additive group instead of $\set{\pm 1}$; see \cite[Section 4.2]{okay2022mermin}.)
{For a maximal non-zero isotropic subspace $I=\set{0,a,b,a+b}$ the non-zero edges can be represented by a triangle $C=\set{a,b,a+b}$ in the Mermin torus.}

\begin{thm}[\!\cite{okay2022mermin}]
\label{thm:MP1 neighbors} 
Let $v,w\in \MP$ be two vertices associated with $(\Omega_v,{\gamma}_v)$ and $(\Omega_w,{\gamma}_w)$.
Then $v$ and $w$ are neighbors if and only if there exists a cnc set $\Omega$ and a function $\varphi:\Omega^c\to \ZZ_2$ such that
$$
w_a = \left\lbrace
\begin{array}{ll}
v_a & a\in \Omega \\
v_a+\frac{(-1)^{\varphi(a)}}{2} & \text{otherwise.}
\end{array}
\right.
$$
Moreover, we have
\begin{enumerate}
\item {$\Omega^c=(\Omega_v-\Omega_w) \cup (\Omega_w-\Omega_v)$,}
%$\Omega = (\Omega_v-\Omega_w) \cup (\Omega_w\cup\Omega_v)^c$,
\item $\varphi(a)={\gamma}_v(a)+1$ for $a\in \Omega_v{\cap \Omega^c}$,
\end{enumerate}
and for a fixed $\Omega$ there are two such functions $\varphi$ corresponding to $\varphi(b)\in \ZZ_2$ for $b\in \Omega^c\cap C$, where $C$ is a triangle with $|C\cap \partial\Omega_v|=1$.
\end{thm}

 \begin{figure}[h!] 
  \centering
  \includegraphics[width=.8\linewidth]{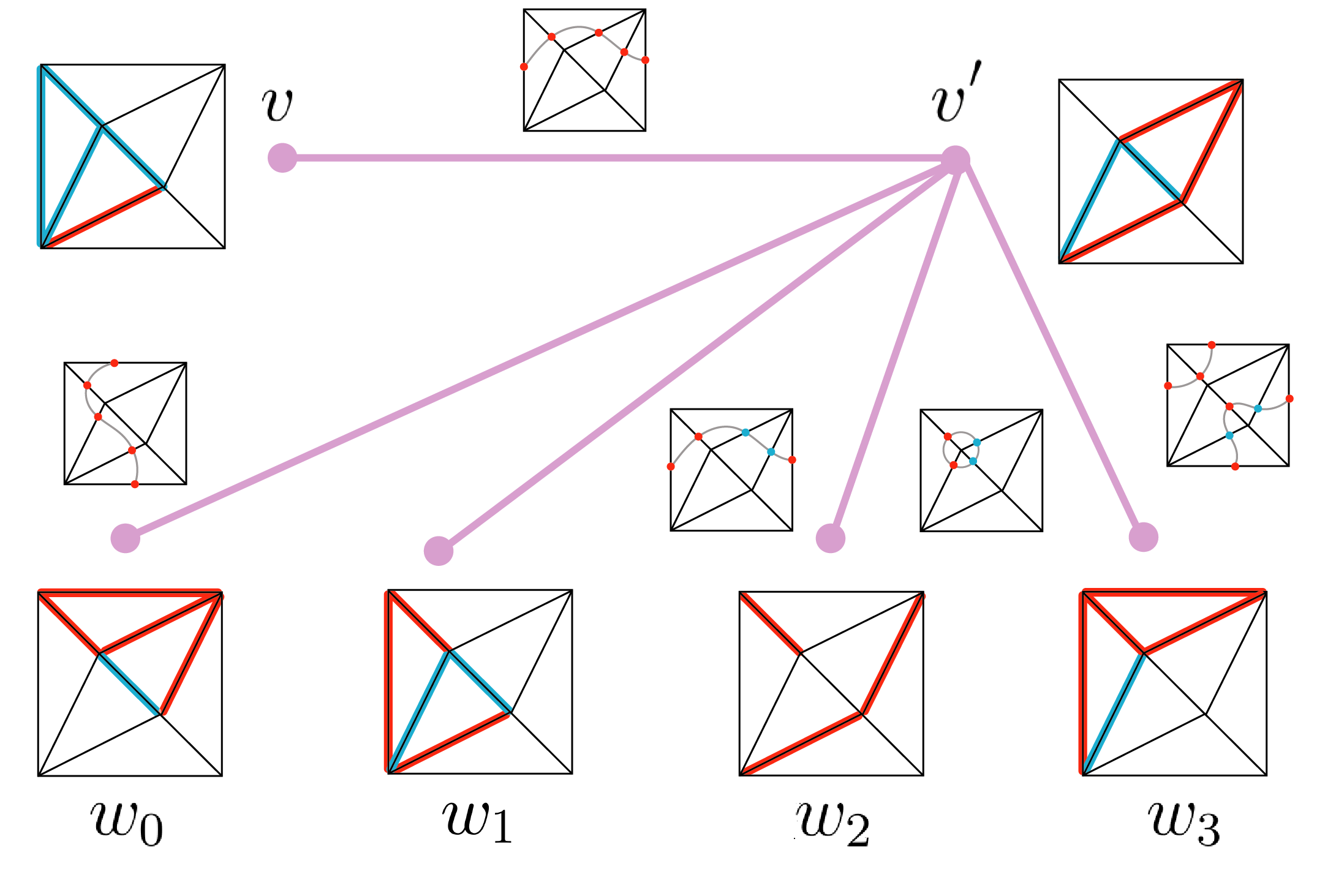} 
\caption{ 
Loops that connect the neighbor vertices.
}
\label{fig:mermin-neighbors}
\end{figure}  
 
{
\Cor{\label{cor:intersection two elements and outcome assignments coincide}
Let $v$ be {a} $\bar T_2$ vertex. For a neighbor $w$ of $v$ we have that $|\partial\Omega_v\cap \Omega_w|=2$. Moreover, ${\gamma}_v$ and ${\gamma}_w$ coincide on $\partial\Omega_v\cap \Omega_w$.    
}
\Proof{
This follows from Theorem \ref{thm:MP1 neighbors}. More specifically, Fig.~(18) in \cite{okay2022mermin} gives a list of the neighbors of the canonical $\bar T_2$ vertex.
} 
}

\Cor{\label{cor:2 neighbors}
The vertices $v$ and $w_i$, where $i=0,1,2,3$, {given in Eq.~(\ref{eq:canonical-vertex}) and Eqns.~(\ref{eq:u0})-(\ref{eq:u3})} of $\MP$ are {not neighbors, but they are common neighbors of a vertex {$v'$}; see Fig.~(\ref{fig:mermin-neighbors}).}}
\Proof{
By Theorem \ref{thm:MP1 neighbors} neighbors are described by loops. {We observe that i}t is not possible to find $\varphi$ satisfying the specified conditions since there exists a point in $\Omega_v\cap \Omega_{w_i}$ at which the values of ${\gamma}_v$ and ${\gamma}_{w_i}$ differ.
}

\subsection{Joint rank}

Consider the matrix ${N}=\bar M^\nloc$
%\comm{$\bar M$ is used for $\overline\MP$}
~whose rows are indexed by non-local stabilizer subgroups and columns are indexed by non-local Pauli operators; see Eq.~(\ref{eq:MP1 polytope}). Since we can parametrize a stabilizer subgroup by a pair $(I,{\gamma})$ consisting of a maximal isotropic subspace $I$ and an outcome assignment ${\gamma}:I\to \ZZ_2$, we will think of the rows being indexed by such pairs. 
%The rows are indexed by $E^\nloc$. 
For $x\in \MP$ the set $\zZ_x$ consists of the {tight} inequalities, i.e., the rows for which the corresponding inequality is an equality:
$$
h_{I,{\gamma}} \cdot x =0
$$
where $h_{I,{\gamma}}\in H(\MP)$. Our goal is to compute
$$
\rank(v,w) = {N}[\zZ_v\cap \zZ_w]
$$
for two vertices $v,w$ of $\MP$. We will follow a graph-theoretic approach. Consider the complete bipartite graph $K_{3,3}$. Its vertices are partitioned into two types: (1) maximal non-local isotopic subspaces and (2) the elements of $E^\nloc$. The edges indicate the containment relation, that is, there is an edge $e_{(a,I)}$ between $a$ and $I$ whenever $a\in I${; see Fig.~(\ref{fig:k33})}

 \begin{figure}[h!] 
  \centering
\includegraphics[width=.25\linewidth]{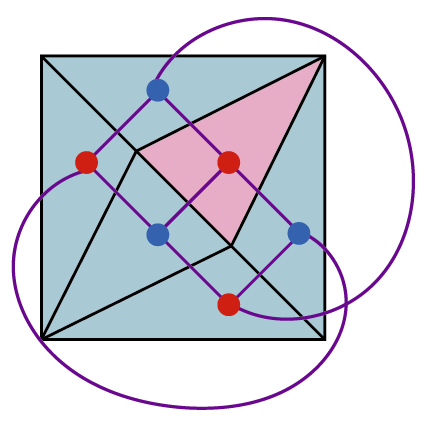} 
\caption{  
}
\label{fig:k33}
\end{figure}  

{The f}irst step in computing the joint rank is to observe that given $v,w$ the joint rank can be computed as the rank of ${N}[\zZ_u]$ where $u=pv+(1-p)w$ {and} $p\in (0,1)$. 
Let  $E_u$ denote the set of $a\in E^\nloc$ such that $u_a\in \set{\pm 1}$. We refer to such elements as deterministic edges.
{Next}, we construct a subgraph $G_u$:
\begin{itemize}
\item The vertex set is given by $V(G_u)=\hat\zZ_u$ where $\hat \zZ_u\subset \zZ_u$ consisting of 
\begin{enumerate}
\item the single inequality for each $I$ with no deterministic edges,
\item choice of one of the inequalities for each $I$ with a single deterministic edge,
\item no inequalities if $I$ involves three deterministic edges.  
\end{enumerate} 

\item The edge set $E(G_u)$ is determined by the subgraph requirement.
\end{itemize}
{Note that if two of the edges of a triangle is deterministic then the third is also deterministic \cite[Lemma 3.9]{okay2022mermin}. Therefore we can only have the three cases in the description of the edges.}

\Lem{\label{lem:separating deterministic part}
We have
$$
\rank({N}[\zZ_u]) = |E_u| + \rank({N}[\hat\zZ_u]).
$$
}
\Proof{
{Let $I=\set{0,a,b,a+b}$ be a maximal non-local isotropic subspace and $C=\set{a,b,a+b}$ the corresponding triangle.
The result follows from the observation that in a triangle $C$ with a single deterministic edge  we have $|\zZ_u\cap C|=2$ and 
$$
\begin{bmatrix}
(-1)^r & (-1)^s & (-1)^{r+s{+}\beta(a,b)}   \\
(-1)^{r+1} & (-1)^{s+1} & (-1)^{r+s{+}\beta(a,b)}
\end{bmatrix} \sim
\begin{bmatrix}
(-1)^r & (-1)^s & (-1)^{r+s{+}\beta(a,b)}   \\
0 & 0& 2(-1)^{r+s{+}\beta(a,b)}
\end{bmatrix}.
$$
Therefore the rank is given by $2$, i.e., one deterministic edge plus one tight inequality.
When all three edges in $C$ are deterministic then a similar calculation shows that the rank is $3$, i.e., three deterministic edges and zero tight inequalities.
}
}

{To conclude our joint rank computation w}e will use some graph-theoretic notions from \cite{zaslavsky2013matrices}.
{We can associate} $G_u$  a bidirection:
$$
\eta(a,h) = h_a
$$
where $h=h_{(I,{\gamma})}\in \hat\zZ_u$ and $a\in {C}$. The bidirection {provides a}  a canonical way to {associate} a sign {with} the edges:
$$
\sigma(a) = -\eta(a,h)\eta(a,h') 
$$
where $h'=h'_{I',{\gamma}}\in \hat\zZ_u$ and $a\in{C\cap C'}$. {Here $C'=I'-\set{0}$.}
This gives a signed graph $G_u^\sigma$. A $2$-regular connected {subgraph} of a graph is called a {\it circle}. If the initial graph is signed then the sign of the circle is obtained by multiplying the sign of its edges.
A signed graph is called {\it balanced} if every circle is positive.  Let $b(G_u)$ denote the number of connected components that are balanced.

\Lem{\label{lem:balanced}
We have
$$
\rank({N}[\hat\zZ_u]) = |\hat \zZ_u|-b(G_u).
$$
}
\Proof{
This follows from \cite[Theorem 4.1]{zaslavsky2013matrices}.
}

\Proof{[{\bf Proof of Proposition \ref{pro:wi}}]
{Corollary} \ref{cor:2 neighbors} shows that $v$ and $w_i$ {are not neighbors, but {both} are  neighbors of a vertex $v'$.}
% are $2$-neighbors. 
The joint rank computation can be done using the following formula based on Lemma \ref{lem:separating deterministic part} and \ref{lem:balanced}:
$$
\rank(v,w_i) =  |{E_{u_i}}|+|\hat \zZ_{u_i}|-b(G_{u_i})
$$
where $u_i=pv+(1-p)w_i$ for $p\in (0,1)$.
In Fig.~(\ref{fig:ranks}) we observe that $b(G_{u_i})=0$ for $i=0,1,2,3$. Therefore the rank is $7$ for $i=0$, and otherwise $6$.
}

\begin{figure}[h!]
\centering
\begin{subfigure}{.25\textwidth}
  \centering
   \includegraphics[width=.8\linewidth]{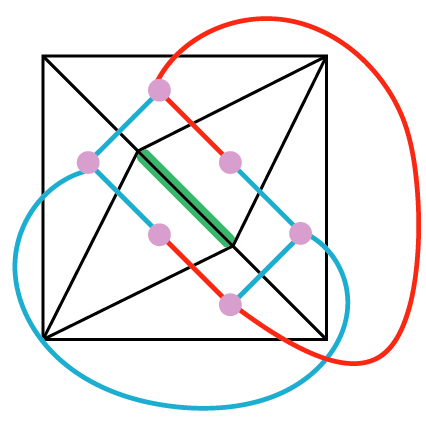}
  \caption{$u_{0}$}
  \label{fig:ranks-u0}
\end{subfigure}%
\begin{subfigure}{.25\textwidth}
  \centering
   \includegraphics[width=.8\linewidth]{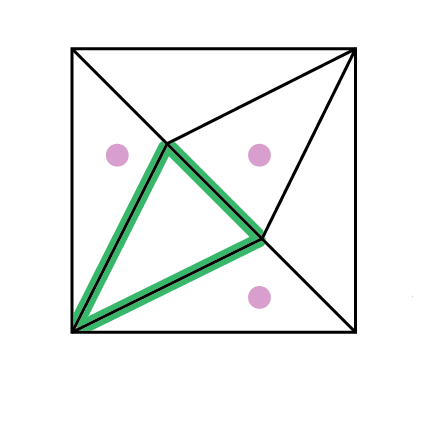}
  \caption{$u_{1}$}
  \label{fig:ranks-u1}
\end{subfigure}%
\begin{subfigure}{.25\textwidth}
  \centering
   \includegraphics[width=.8\linewidth]{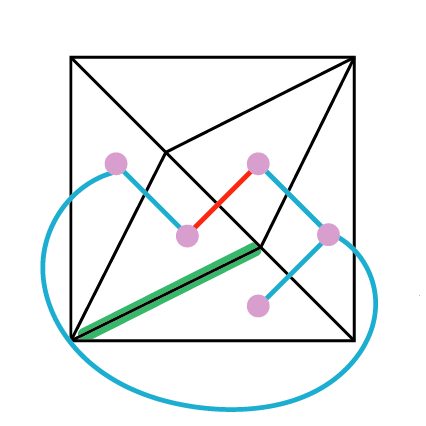}
 \caption{$u_{2}$}
  \label{fig:ranks-u2}
\end{subfigure}%
\begin{subfigure}{.25\textwidth}
  \centering
   \includegraphics[width=.8\linewidth]{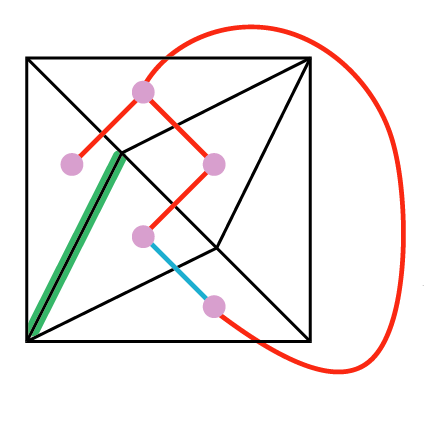}
  \caption{$u_{3}$}
  \label{fig:ranks-u2}
\end{subfigure}
\caption{The signed graph $G_{u_i}^\sigma$ for $i=0,1,2,3$. Colors on the edges of the graph indicate the sign: blue for $1$ and red for $-1$. Green color indicates the deterministic edges in the torus.  
}
\label{fig:ranks}
\end{figure}

\section{Proof of Lemma~\ref{lem:tight-ui}}\label{sec:proof-Lem-tight}
 
We seek all $h\in H_{2}$ such that $h\cdot u_{i} = 0$ ($i=0,1,2,3$). Table~\ref{tab:tight-cube-ui} can be obtained straightforwardly by reading off the components of $u_{i}$ whose values are $\pm 1$. To obtain the tight CHSH inequalities we use the equation $h_{(\partial\Omega,\gamma)}\cdot u_{i} = 0$ and run over all possible pairs $(\partial\Omega,\gamma)$. Note that for all $u_{i}$ there exists some $a\in E^{\nloc}$ such that $(u_{i})_{a} = \pm 1$. Thus it is useful to distinguish the two cases (i) $a\in \partial\Omega$ and (ii) $a\notin \partial\Omega$. Representatives of these two cases for $u_{0}$ are shown below:
\begin{eqnarray}
\includegraphics[width=0.6\linewidth]{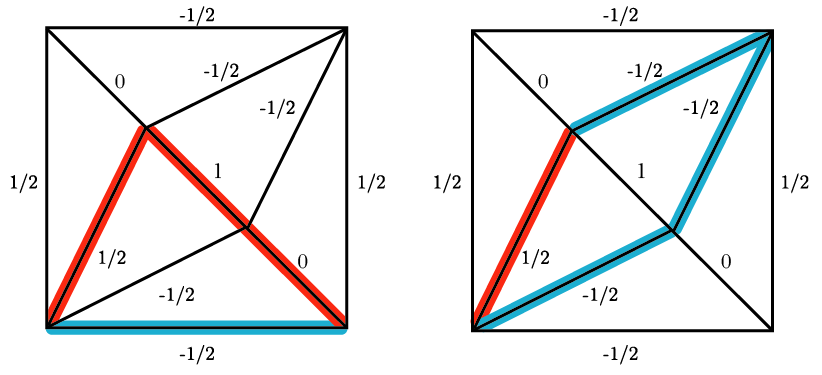}
\end{eqnarray}
Let us sketch how we obtain these inequalities. If $a \in \partial\Omega$ then there are four possible sets $\partial\Omega$ and $\gamma$ must be such that $\gamma(a) = 1$ and the remaining terms must equal $-1$. For $u_{0}$ this can be arranged for all four sets. For the case where $a\notin \partial\Omega$ we need $\partial \Omega$ such that all components $(u_{i})_{b} = \pm 1/2$, where $b\in \partial\Omega$. There are two such instances. The calculations for the other vertices $u_{i}$ follow analogously. Note that by construction all $u_{i}$ must be tight at $h_{v}\cdot x\geq 0$, where $v\in \MP$ is the canonical vertex in Eq.~(\ref{eq:canonical-vertex}).

\section{Proof of Proposition~\ref{pro:decomposition}}\label{sec:proof-Pro-decomp}

Satisfying the NN inequality with equality forces all {$\delta^\xi$} in Eq.~(\ref{eq:convex-bar_C}) to have $({\delta^\xi})_{a} = \pm 1$ whenever $(u_{i})_{a} = \pm 1$. Every such NN inequality divides the number of prospective {$\delta^\xi$}'s in half.\footnote{For instance, $u_{1}$ satisfies 
{$(u_1)_{ZZ}=(u_1)_{XX}=-(u_1)_{YY}=1$,}
%$\Span{ZZ}=\Span{XX}=-\Span{YY}=1$, 
yielding precisely $4 = 2^{5}/2^{3}$ possible $\delta^\xi$'s. The remaining CHSH inequalities impose no further restrictions.} The remaining $m_{i}\in \mathbb{N}$ CHSH inequalities can be organized into a matrix ${\bar N^{(i)}}\in \RR^{m_{i}\times 10}$ whose rows are given by the vectors $h\in \RR^{10}$ that define a CHSH inequality. Our goal is to find all $\delta^\xi \in V_{2}$ that satisfy ${\bar N^{(i)}}\delta^\xi = \zero_{m_{i}\times 1}$. Since $V_{2}$ is known this amounts to the straightforward (but tedious) combinatorial task of checking which $\delta^\xi$'s satisfy these conditions.\footnote{This task is equivalent to computing the (partial) face lattice (see e.g., \cite{ziegler2012lectures}) of $\overline \CL$.} This yields Table~\ref{tab:tight-vertices}.

We demonstrate how to obtain the non-negative parameters for $u_{0}$. The cases $u_{1}$ and $u_{3}$ follow analogously. We will discuss $u_{2}$ at the end. We seek a solution to Eq.~(\ref{eq:convex-decomposition-a}) where $R$ is given in Table~\ref{tab:R-u0}. There are four unknowns, which we denote as $q_{1},\cdots,q_{4}$, in ten equations and it suffices to isolate a subset of four linearly independent equations from this set:
\begin{eqnarray}
1 &=& q_{1}+q_{2}+q_{3}+q_{4}\notag\\
0 &=& q_{1}+q_{2}-q_{3}-q_{4}\notag\\
0 &=& q_{1}-q_{2}+q_{3}-q_{4}\notag\\
\frac{1}{2} &=& q_{1}+q_{2}+q_{3}+q_{4}\notag.
\end{eqnarray}
These come from the constraints Eq.~(\ref{eq:convex-decomposition-a}) with Pauli components $II$, $XY$, $XZ$, and $YX$. Solving these equations yields $q_{1}=q_{2}=q_{3}=q_{4}=1/4$. Should the solution not have been non-negative (which is generally not guaranteed) we could have implemented a two-phase linear programming approach \cite{luenberger1984linear} to find such a solution. For our purposes it, in fact, suffices to find {any} decomposition at all. One can check that the coefficients given in Eq.~(\ref{eq:decomp-u2}) with $R_{2}^{(2)}$ yield $u_{2}$.

\end{document}